\documentclass[a4paper]{article}

\usepackage[left=3cm,top=2.5cm,right=3cm,bottom=2.5cm]{geometry}

\title{$J$-states and quantum channels between indefinite metric spaces}
\author{ {\bf{ Ra\'{u}l Felipe-Sosa$^{a}$ and Ra\'ul Felipe$^{b}$}}\\ \\
        $^{a}$ BUAP, Puebla. M\'{e}xico. \\
        $^{b}$CIMAT \\
        Callej\'on Jalisco s/n Mineral de Valenciana \\
        Guanajuato, Gto, M\'exico.\\
        raulf@cimat.mx}
\date{}
\usepackage{graphicx,array,delarray}
\usepackage{latexsym}
\usepackage{amscd}
\usepackage{amsfonts, amsmath, amssymb}
\newtheorem{theorem}{\textbf{Theorem}}

\newtheorem{corollary}[theorem]{\textbf{Corollary}}

\newtheorem{example}{\textbf{Example}}
\newtheorem{definition}[theorem]{\textbf{Definition}}

\newtheorem{lemma}[theorem]{\textbf{Lemma}}

\newtheorem{proposition}[theorem]{\textbf{Proposition}}

\newtheorem{remark}[theorem]{Remark}

\newenvironment{proof}[1][Proof]{\noindent\textbf{#1.} }{\ \rule{0.5em}{0.5em}}

\begin{document}

\maketitle

\begin{abstract}
In the present work, we introduce and study the concepts of state and quantum channel on spaces equipped with an indefinite metric.
Exclusively, we will limit our analysis to the matricial framework. As it will be confirmed below, from our research it is noticed that,
when passing to the spaces with indefinite metric, the use of the adjoint of a matrix with respect to the indefinite metric is required in the construction of states and quantum channels; which prevents us to consider the space of matrices of certain order $M_{n}(\mathbb{C})$ as
a $C^{\ast}$-algebra. In our case, this adjoint is defined through a $J$-metric, where the matrix $J$ is a fundamental symmetry
of $M_{n}(\mathbb{C})$. In our paper, for quantum operators, we include the general setting in the which, these operators map
$J_{1}$-states into $J_{2}$-states, where $J_{2}\neq \pm J_{1}$ are two arbitrary fundamental symmetries. In the middle of this program, we carry
out a study of the completely positive maps between two different positive matrices spaces by considering two different indefinite metrics on $\mathbb{C}^{n}$.
\end{abstract}

\noindent{\it{2010 Mathematics Subject Classification (MSC2010):}}
Primary 81P16, 81P45; Secondary 15B48, 46C20.\\


\section{Introduction}

Quantum channels are the entities through which the encoded information of quantum systems is transmitted in the form of states.
All these objects have their mathematical interpretation, in which the theory of completely positive maps stands out.
Thus, quantum states, completely positive maps and quantum channels are basic tools of the functional analysis underlying in the quantum
information theory.

The aim of the present paper is to provide a more general setting for the theory of quantum information by means of tools arising
from operators theory on indefinite metric spaces. In particular, we introduce the notion of $J$-state, completely $J$-positive map
and quantum $J$-channel, where $J$ is a fundamental symmetry matrix of $M_{n}(\mathbb{C})$, that is, $J^{\ast}=J$ and $J^{2}=J$. We consider the
general situation in which the transmission of information is carried out between two different indefinite metric spaces with the same
underlying vector space (see section $4$). The results of the article generalize those of the case $J=I_{n}$, where $I_{n}$ is the identity
matrix of order $n$.

From a mathematical point of view, any of the problems that are studied for the usual quantum channels are susceptible of being transferred
to the new situation considered in this article.

The paper is organized as follows. In section $2$, we introduce and study the notion of $J$-state which will derive in an analysis
about the convenience or not of requiring or modifying the condition $Tr\,A=1$ for a $J$-state. In section $3$, we introduce and study the new
concepts of completely $J$-positive maps and quantum $J$-channels in which the Kraus type operators are revisited; this fact is
characterized by the extensive use of the $J$-adjoint of a matrix. In section $4$, we also consider quantum channels which transform
$J_{1}$-states in $J_{2}$-states belong to $M_{n}(\mathbb{C})$, here $J_{2}\neq \pm J_{1}$ are two fundamental symmetries. Theses kinds of
channels are called by us quantum $(J_{1},J_{2})$-channels.

\subsection{Motivation, indefinite quantum mechanics, the $J$-Bloch ball}

Early last century physicists$'$s need of computing probabilities of events concerning the subatomic world
of elementary particles led to the arise of quantum probability (the interested reader can consult \cite{partha}). On the other hand, in the
last century, there were some efforts to transfer quantum mechanics to spaces with indefinite metric which started with the work of P. M. Dirac, one
of the founders of this type of spaces together with S.L. Sobolev in the former Soviet Union (Sobolev discovered the Pontryagin spaces $\Pi_{\kappa}$
of order $1$, that is for $\kappa=1$, in the case that, the indefinite metric is defined on a underlying vector spaces of infinite dimension).
Some of these articles pointed out the possibility of using negative values for quantum probabilities (see \cite{ascoli}, \cite{gupta1},
\cite{miyatake}, \cite{schnitzer},). Below, we will try to approach our work in this context.

Let $p=(p_{1},p_{2})$ be a vector such that $p_{1}> 0$, $p_{2}< 0$ and $p_{1}-p_{2}=1$. Then, we do the following observations:
\begin{enumerate}
  \item It is clear that the vector $p_{1,1}=(p_{1},-p_{2})$ is a distribution of probability on the finite probability space $\{1,2\}$. Moreover,
        the vector $s_{p}=(\sqrt{p_{1}},\sqrt{-p_{2}})$ is a usual state vector of $\mathbb{C}^{2}$. Note that
        \begin{equation}\label{for weinsten1}
        p_{1}-p_{2}=Tr\,\left(
                          \begin{array}{cc}
                            1 & 0 \\
                            0 & -1 \\
                          \end{array}
                        \right)\left(
                                 \begin{array}{cc}
                                   p_{1} & 0 \\
                                   0 & p_{2} \\
                                 \end{array}
                               \right)
        =Tr\,JP_{J}=1,
        \end{equation}
        the matrix $P_{J}$ will be called $J$-density matrix, in this case. Note that $JP_{J}$ is a positive definite matrix for which
        $Tr\,JP_{J}=1$.
  \item Let $A=(a_{1},a_{2})$ be an observable of phase space $\mathbb{C}^{2}$ which is the result of a measurement, thus
  $A\in \mathbb{R}^{2}$. The generalized probability distribution $p$ is associated to $A$ which is apparent in the usual metric of $\mathbb{C}^{2}$
  but effective with respect to the indefinite metric induced by the matrix $J$. Then, the $J$-expectation of $A$ is, in this case
  \begin{equation}\label{for weinsten2}
  E_{J}(A)=p_{1}a_{1}-p_{2}a_{2}=Tr\,\left(
                          \begin{array}{cc}
                            1 & 0 \\
                            0 & -1 \\
                          \end{array}
                        \right)\left(
                                 \begin{array}{cc}
                                   p_{1} & 0 \\
                                   0 & p_{2} \\
                                 \end{array}
                               \right)\left(
                                        \begin{array}{cc}
                                          a_{1} & 0 \\
                                          0 & a_{2} \\
                                        \end{array}
                                      \right)=Tr\,P_{J}JA\,.
  \end{equation}
\end{enumerate}

Taking into account that the Bloch ball is a well known model of quantum mechanics, we will try to replicate this with respect to a certain indefinite
metric on $\mathbb{C}^{2}$. Let us stay in the space of matrices of order $2$ with complex entries which is denoted by $M_{2}(\mathbb{C})$.
Any matrix $S\in M_{2}(\mathbb{C})$ is a linear combination of the identity matrix $I_{2}$ and the Pauli matrices
\begin{equation*}
\sigma_{1}=\left(
             \begin{array}{cc}
               0 & 1 \\
               1 & 0 \\
             \end{array}
           \right),\,\,\,\,\,\,\,\,\,\,\,\,\,\sigma_{2}=\left(
                                                  \begin{array}{cc}
                                                    0 & -i \\
                                                    i & 0 \\
                                                  \end{array}
                                                \right),\,\,\,\,\,\,\,\,\,\,\,\,\,\sigma_{3}=\left(
                                                                                               \begin{array}{cc}
                                                                                                 1 & 0 \\
                                                                                                 0 & -1 \\
                                                                                               \end{array}
                                                                                             \right),
\end{equation*}
that is, $S=\frac{1}{2}\left [z_{0}I_{2}+z_{1}\sigma_{1}+z_{2}\sigma_{2}+z_{3}\sigma_{3} \right ]$, where $(z_{1},z_{2},z_{3},z_{4})\in \mathbb{C}^{4}$.
The crucial fact of this decomposition is that $Tr\,S=z_{0}$.

Denote by $J$ the following signature matrix $J=\left(
                                                \begin{array}{cc}
                                                 1 & 0 \\
                                                 0 & -1 \\
                                                 \end{array}
                                                 \right)$ which satisfies the following properties $J^{\ast}=J$ and
$J^{2}=I_{2}$. Next, we will consider the indefinite metric space $(\mathbb{C}^{2}, [.,.])$ where $[x,y]=\langle Jx,y\rangle$.
Then, it is easy to see that
\begin{equation}\label{editor3}
JS=\frac{1}{2}\left [z_{3}\left(
                            \begin{array}{cc}
                              1 & 0 \\
                              0 & 1 \\
                            \end{array}
                          \right)-iz_{2}\left(
                                          \begin{array}{cc}
                                            0 & 1 \\
                                            1 & 0 \\
                                          \end{array}
                                        \right)+iz_{1}\left(
                                                        \begin{array}{cc}
                                                          0 & -i \\
                                                          i & 0 \\
                                                        \end{array}
                                                      \right)+z_{0}\left(
                                                                     \begin{array}{cc}
                                                                       1 & 0 \\
                                                                       0 & -1 \\
                                                                     \end{array}
                                                                   \right)
 \right],
\end{equation}
thus $Tr\,JS=z_{3}$. Therefore, the next observations follow from (\ref{editor3})\,:
\begin{itemize}
  \item The matrix $JS$ is self-adjoint if and only if $z_{3}=x_{3}$ and $z_{0}=x_{0}$ are real, moreover, $z_{2}=ix_{2}$ and $z_{1}=-ix_{1}$ being $x_{1}, x_{2}$ both real. In other words, every $J$-selfadjoint matrix $S$ has the form
      \begin{equation}\label{editor31}
      S=\left(
          \begin{array}{cc}
            \frac{x_{3}+x_{0}}{2} & \frac{x_{2}-ix_{1}}{2} \\
            -\frac{x_{2}+ix_{1}}{2} & -\frac{x_{3}-x_{0}}{2} \\
          \end{array}
        \right).
      \end{equation}
  \item Now, $JS$ is a positive definite matrix if and only if $\|(x_0,x_{1},x_{2})\|_{\mathbb{R}^{3}}\leq x_{3}$. This happens because the characteristic polynomial of $JS$ is, in our case, $\lambda^{2}-x_{3}\lambda+\frac{x^{2}_{3}-x^{2}_{2}-x^{2}_{1}-x^{2}_{0}}{4}$, which will have two positive eigenvalues under the previous restriction.
\end{itemize}

One can ensure that if $z_{3}=x_{3}=1$, then, $Tr JS=1$ and we have $0\leq JS$ provided that $\|(x_0,x_{1},x_{2})\|_{\mathbb{R}^{3}}\leq 1$.
We call this set, the \textbf{Bloch $J$-ball}. The analysis for $J=I_{2}$ can be found in \cite{wolf}.
On the other hand, every $2\times 2$ matrix $A$ such that $A=JA^{\ast}J=A^{\natural}$ (this means that $A$ is $J$-selfadjoint) has the form
\begin{equation}\label{editor32}
A=\left(
          \begin{array}{cc}
            \frac{y_{0}+y_{3}}{2} & \frac{y_{2}+iy_{1}}{2} \\
            -\frac{y_{2}-iy_{1}}{2} & \frac{y_{0}-y_{3}}{2} \\
          \end{array}
        \right),
\end{equation}
where the $y_{i}$ are real for $i=1,2,3,4$. \\

On the stage $(M_{2}(\mathbb{C}),J)$, we define the following notions\,:
\begin{itemize}
  \item A $J$-observable is a matrix $A$ of type (\ref{editor32}), such that, each $y_{i}$ is real for all $i$. Since $A$ is a $J$-selfadjoint
        matrix on $\Pi_{1}=(\mathbb{C}^{2},J)$, the spectrum of $A$ is symmetric with respect to the real axis. The more general result
        in this sense, can be found in \cite{krein}, page  $135$. \textbf{Then, the possible values of $A$ (its outcomes) are defined as the
        real part of the eigenvalues}.

  Observe that in the classic quantum mechanics in which $J$ reduces to the identity matrix $I_{2}$, $A$ is
  a usual selfadjoint matrix, so the spectrum of $A$ is real, which is in accordance with the previously given definition.
  \item A \textbf{quantum $J$-state} is represented by a matrix of the form
  \begin{equation}\label{editor4}
  S=\left(
          \begin{array}{cc}
            \frac{1+x_{0}}{2} & \frac{x_{2}-ix_{1}}{2} \\
            -\frac{x_{2}+ix_{1}}{2} & -\frac{1-x_{0}}{2} \\
          \end{array}
        \right),
  \end{equation}
  where $\|(x_{0},x_{1},x_{2})\|_{\mathbb{R}^{3}}\leq 1$. Here, the variables $x_{i}$ denote real numbers for $i=0,1,2$. The eigenvalues of $S$ are necessarily real (the proof that both eigenvalues are on the real axis follows from theorem $3.27$ page $111$ of \cite{azizov}). Even more, we
  can calculate exactly these eigenvalues
  \begin{equation}\label{editor5}
  p_{1}(S)=\frac{x_{0}+1}{2},\,\,\,\,\,\,\,\,\,\, p_{2}(S)=\frac{x_{0}-1}{2}=-\frac{1-x_{0}}{2},
  \end{equation}
  and taking into account that $|x_{0}|< 1$, we can assume that $0< x_{0} < 1$ corresponding to case (\ref{for weinsten1}). In fact, it is clear
  that $p_{1}-p_{2}=1$. We say that $p(S)=(p_{1}(S),p_{2}(S))$ is the \textbf{generalized probability distribution} corresponding to $S$.
  \end{itemize}

Now, we proceed as in \cite{partha}. With this purpose, we return to the $J$-observable matrices (\ref{editor32}). Let us recall that $x_{0}, x_{1}, x_{2}, x_{3}\in \mathbb{R}$. Suppose next
$y^{2}_{1}+y^{2}_{2}< y^{2}_{3}$ further $y_{1}\neq 0$ or $y_{2}\neq 0$. Then, the eigenvalues of $A$, that is, its possible values, are real
and they have the following form (the calculations that follow were done with MATLAB)
\begin{equation}\label{editor6}
\lambda_{1}=\frac{y_{0}+\sqrt{y^{2}_{3}-y^{2}_{1}-y^{2}_{1}}}{2},\,\,\,\,\,\,\,\,\,\,\,\,\lambda_{2}=\frac{y_{0}-\sqrt{y^{2}_{3}-y^{2}_{1}-y^{2}_{1}}}{2},
\end{equation}
and its corresponding eigenvectors are
\begin{equation}\label{editor7}
  E_{1}=\left(
          \begin{array}{c}
            -\frac{y_{3}+\sqrt{y^{2}_{3}-y^{2}_{1}-y^{2}_{1}}}{y_{2}-iy_{1}} \\
            1 \\
          \end{array}
        \right),\,\,\,\,\,\,\,\,\,\,\,\,\,\,E_{2}=\left(
          \begin{array}{c}
            \frac{y_{3}-\sqrt{y^{2}_{3}-y^{2}_{1}-y^{2}_{1}}}{y_{2}-iy_{1}} \\
            1 \\
          \end{array}
        \right),
\end{equation}
which are two linearly independent vectors of $\mathbb{C}^{2}$ because
\begin{equation*}
|E_{1}\,\,E_{2}|=\frac{2\sqrt{y^{2}_{3}-y^{2}_{1}-y^{2}_{1}}}{y_{2}-iy_{1}}\neq 0,
\end{equation*}
it shows that $\mathbb{C}^{2}=\langle E_{1}\rangle\oplus \langle E_{2}\rangle$, where $\langle E_{k}\rangle$ is the subspace spanned by $E_{k}$
for $k=1,2$. On the other hand, since $\lambda_{1}\neq \lambda_{2}$ we have $[E_{1},E_{2}]=\langle JE_{1},E_{2} \rangle=0$, which implies that even
more $\mathbb{C}^{2}=\langle E_{1}\rangle[\dotplus] \langle E_{2}\rangle$. Here, the symbol $[\dotplus]$ means $J$-orthogonal direct sum. Define $V_{k}=\frac{E_{k}}{\|E_{k}\|}$ for $k=1,2$. Then, any vector $V\in \mathbb{C}^{2}$ can be written in the form $V=[.\,, V_{1}]V_{1}+[.\,, V_{2}]V_{2}$.
It shows that
\begin{equation}\label{editor8}
A\,V=\lambda_{1}[V\,, V_{1}]V_{1}+\lambda_{2}[V\,, V_{2}]V_{2}=\lambda_{1}\Pi_{1}V+\lambda_{2}\Pi_{2}V,
\end{equation}
for any $V\in \mathbb{C}^{2}$, where we have denoted $\Pi_{k}=[.\,, V_{k}]V_{k}=\langle J\,.\,, V_{k}\rangle V_{k}$
which is a $J$-orthogonal projector for each $k$, that is, $\Pi^{\natural}_{k}=J\Pi^{\ast}_{k}J=\Pi_{k}$ and $\Pi_{k} J\Pi_{k}=\Pi_{k}$.

\begin{definition}  Let us denote $H_{k}=\Pi_{k}\mathbb{C}^{2}$. The probability of observing a system in a subspace $H_{k}$ is defined as
$Tr\,(S\Pi_{k})$. The expectation of the $J$-observable $A$ relative to $S$ is $E(A)=Tr\, SA=\lambda_{1}Tr\,S\Pi_{1}+\lambda_{2}Tr\,S\Pi_{2}$.
It means that a single measurement of $A$ in state $S$ produces a values $\lambda_{k}$ with probability $Tr\,(S\Pi_{k})$. We call the pair
$(\mathbb{C}^{2}, S)$ a quantum probability $J$-space.
\end{definition}

\begin{remark}Consider now a $J$-unitary matrix $V$ belongs to $M_{2}(\mathbb{C})$, that is, a matrix for which $VV^{\natural}=V^{\natural}V=I_{2}$.
We recall again that $V^{\natural}=JV^{\ast}J$. This condition is equivalent to the following $V^{\ast}JV=VJV^{\ast}=J$. The $J$-unitary matrices
of $M_{2}(\mathbb{C})$ have the form
\begin{equation*}
V=\left(
    \begin{array}{cc}
      \alpha & \beta \\
      \overline{\beta} & \overline{\alpha} \\
    \end{array}
  \right),\,\,\,\,\,\,\,\,\,\,\,\,\,\,|\alpha|^{2}-|\beta|^{2}=1.
\end{equation*}

It is easy to show that if $A$ is a $J$-observable and $S$ is a $J$-state, then $VAV^{\natural}$ (respectively $V^{\natural}AV$) is a $J$-observable and $VSV^{\natural}$ (respectively $V^{\natural}SV$) is a $J$-state. Hence, it is convenient (and in occasions necessary. For example, when measurements cannot be made) to study the dynamic of the following matrix functions:
\begin{equation}\label{editor9}
A(t)=V_{1}(t)AV_{1}^{\natural}(t),\,\,\,\,\,\,\,\,\,\,\,S(t)=V_{2}(t)SV_{2}^{\natural}(t),
\end{equation}
where $V_{1}(t)$ and $V_{2}(t)$ are families of $J$-unitary matrices. It is well known that (\ref{editor9}) is equivalent to the equations of Lax
type
\begin{equation}\label{editor10}
\frac{dA}{dt}=[B_{1}(t),A(t)],\,\,\,\,\,\,\,\,\,\,\frac{dS}{dt}=[B_{2}(t),S(t)],
\end{equation}
where $B_{1}(t)=-V_{1}(t)\frac{dV^{\natural}_{1}(t)}{dt}$ and $B_{2}(t)=-V_{2}(t)\frac{dV^{\natural}_{2}(t)}{dt}$ are both skew $J$-selfadjoint
matrices, that is, $B^{\natural}_{k}(t)=-B_{k}(t)$ for $k=1,2$. We can take, as usual, $V_{1}(t)=V_{2}(t)=e^{-itM}$ where in our context $M$ is a
$J$-observable matrix, that is, $M^{\natural}=M$ of the form (\ref{editor32}).
\end{remark}

To finish this subsection, we wish to note that the theory of linear operators on spaces with indefinite metric is not restricted to
the habitual theory of linear operators on Hilbert spaces. This is revealed in the differences that appear, for instance, in the
spectral properties of the operators acting on these spaces. We show an example: it is well known that a selfadjoint operator
has all its eigenvalues on the real axis, while some conjugated pairs of eigenvalues of a $J$-selfadjoint operator can escape from
the real axis. The reader may consult \cite{azizov}, \cite{bognar} and \cite{krein} for more details.

\section{Overview on $J$-indefinite linear algebra. The notion of $J$-state}

In all quantum system, a state describes the current condition of that system. For instance, the states are relevant to study any important
quantum information experiment. In this section, we introduce a generalization of this notion from the point of view of the theory of indefinite metric spaces. Here, these new introduced states are called by us $J$-states and they live in space of $J$-positive matrices where $J$ is a fundamental symmetry of $M_{n}(\mathbb{C})$. Briefly, ew expose the bases of our indefinite quantum proposal presented in this section:
\begin{itemize}
  \item Our quantum system is described by $\mathbb{C}^{n}$ equipped with one or two indefinite metrics of the following form $[.,.]=\langle J.\,,. \rangle_{\mathbb{C}^{n}}$, where $J$ is a fundamental symmetry matrix (see below for details) and $\langle .\,,. \rangle_{\mathbb{C}^{n}}$ is the usual inner product in $\mathbb{C}^{n}$. This system shall be called \textbf{indefinite quantum system}.
  \item In our approach, the quantum states will be introduced following an analogous procedure to that of operator theory in spaces with indefinite metric (see \cite{azizov} and \cite{gohberg}). Specifically, in what follows, a $J$-state will be a matrix $B$ such that $JB$ is a usual state or density matrix in quantum mechanics.
\end{itemize}

The objective of this section is to introduce and study the notion of quantum state in indefinite metric spaces. According to the opinion
of the authors, the relevance of the present study is given by its usefulness in the security of the transmission of information through
quantum $J$-channels.  In this sense, we believe that send the information by a $J$-channel will provide bigger chances to encrypt
the information. \\

Next, we review some aspects of indefinite linear algebra.
We say that an indefinite inner product (or metric) is given in $\mathbb{C}^{n}$ if additionally to the usual inner product $\langle \cdot,\cdot \rangle_{\mathbb{C}^{n}}$ we have a function $[\cdot,\cdot]$ from $\mathbb{C}^{n}\times \mathbb{C}^{n}\longrightarrow \mathbb{C}$, which
satisfies the following axioms:
\begin{enumerate}
  \item $[\alpha x_{1}+\beta x_{2},y]=\alpha[x_{1},y]+\beta[x_{2},y]$,\,\,\,\,\,\,$\forall\,x_{1},x_{2},y\in \mathbb{C}^{n}$,\,\,\,\,\,\,
        $\forall\,\,\alpha,\beta \in\mathbb{C}$,
  \item $[x,y]=\overline{[y,x]}$,\,\,\,\,$\forall\,x,y\in \mathbb{C}^{n}$,
  \item if\,\,\,$[x,y]=0$,\,\,\,\,$\forall\,y\in \mathbb{C}^{n}$\,\,\,then\,\,\,\,$x=0$.
\end{enumerate}

One can check that all invertible $n\times n$ selfadjoint matrix $H$ induces an indefinite metric, through the formula
\begin{equation*}
[x,y]=[x,y]_{H}=\langle Hx,y \rangle_{\mathbb{C}^{n}},\,\,\,\,\,for\,\,\,all\,\,\,x,y\in \mathbb{C}^{n}.
\end{equation*}

In this paper, we will only concentrate in a special case of this type of indefinite metrics, when $H=J\in M_{n}(\mathbb{C})$ is a fundamental
symmetry matrix, which means that $J^{\ast}=J$ and $J^{2}=I_{n}$. One can easily exhibit the structure of every fundamental symmetry $J$. Define $P_{+}=\frac{I_{n}+J}{2}$, then $P^{\ast}_{+}=P_{+}$
and $P^{2}_{+}=P_{+}$, that is $P_{+}$ is an orthogonal projection matrix. Put $P_{-}=I_{n}-P_{+}=\frac{I_{n}-J}{2}$. It shows that $J=P_{+}-P_{-}$ with $P_{+}+P_{-}=I_{n}$. Conversely, suppose that $P,Q$ are two orthogonal projection matrices, such that, $P+Q=I_{n}$ then $J=P-Q$ is a fundamental symmetry.

We recall that the matrix $M^{\natural}$ is the $J$-adjoint of a matrix $M$ if by definition $[Mx,y]=[x,M^{\natural}y]$ for all
$x,y\in \mathbb{C}^{n}$ where $[x,y]=\langle Jx,y\rangle_{\mathbb{C}^{n}}$. The $J$-adjoint $M^{\natural}$ of $M$ is unique, even more $M^{\natural}=JM^{\ast}J$ for all $M\in M_{n}(\mathbb{C})$, where $M^{\ast}$ stands the usual adjoint of $M$. It is easy to see that
$(M^{\ast})^{\natural}=(M^{\natural})^{\ast}$.

A matrix $A$ is said to be $J$-positive, if $0\leq [Ax,x]$ for all $x\in \mathbb{C}^{n}$, in which case we have $A^{\natural}=A$. Next,
the set of all $J$-positive matrices is denoted by $M^{+}_{n}(\mathbb{C})(J)$. Notice that if $A$ is a $J$-positive matrix, then, it implies
that $JA$ is a positive matrix in the usual sense. An important property of the $J$-adjoint of matrices is the
following: $(AB)^{\natural}=B^{\natural}A^{\natural}$ which is easy of to prove. Indeed, $(AB)^{\natural}=J(AB)^{\ast}J=JB^{\ast}A^{\ast}J=JB^{\ast}J^{2}A^{\ast}J=B^{\natural}A^{\natural}$. For more on the indefinite metric spaces,
see \cite{azizov} and \cite{gohberg}.

Suppose that the matrix $A$ is $J$-positive, which as was mentioned before, it is equivalent to the fact that $JA$ is a positive matrix
(so self-adjoint and of trace class). Then,
\begin{equation}\label{s1 1}
JA=\sum_{i=1}^{n}\lambda_{i}\langle \cdot, e_{i}\rangle_{\mathbb{C}^{n}} \,e_{i}=\sum_{i=1}^{n}\lambda_{i}\,e_{i}\otimes e_{i},\,\,\,\,\,\,\,\,\,\,\,\,Tr\,JA=\sum_{i=1}^{n}\lambda_{i},
\end{equation}
where some of the $0\leq \lambda_{i}$ could be zero, and $\{e_{i}\}$ is an orthonormal basis of $\mathbb{C}^{n}$. That is,
\begin{equation}\label{s1 2}
I_{n}=\sum_{i=1}^{n}\langle \cdot, e_{i}\rangle_{\mathbb{C}^{n}}\,e_{i},\,\,\,\,\,\,\,\,\langle e_{i},e_{j} \rangle_{\mathbb{C}^{n}}=\delta_{ij},
\end{equation}
and, therefore, from (\ref{s1 1}) it follows that
\begin{equation}\label{s1 3}
A=\sum_{i=1}^{n}\lambda_{i}\langle \cdot, e_{i}\rangle_{\mathbb{C}^{n}} \,Je_{i}=\sum_{i=1}^{n}\lambda_{i}\,(Je_{i})\otimes e_{i}.
\end{equation}

Thus, all $J$-positive matrix $A$ can be written in the form (\ref{s1 3}), where $\{\lambda_{i}\}\subset\mathbb{R}_{+}$ and $\{e_{i}\}$ constitutes an
orthonormal basis. From (\ref{s1 1}) and (\ref{s1 2}) it follows that the first equation of (\ref{s1 1}) represents the spectral decomposition of $B=JA$.
Since $B^{\ast}=B$, then we obtain
\begin{equation}\label{s1 4}
Tr\, B^{\natural}=Tr\,JB^{\ast}J=Tr\,JBJ=Tr\,B=\sum_{i=1}^{n}\lambda_{i}.
\end{equation}

\begin{proposition}\label{proposicion1 s1}
A matrix $A$ is $J$-positive if and only if $A^{\natural}=A$ and there exists $B$ such that $A=B^{\natural}JB$.
\end{proposition}
\begin{proof}Suppose first that $A$ is $J$-positive, then $JA$ is a positive matrix and so self-adjoint. Thus, $JA=(JA)^{\ast}=A^{\ast}J$
which implies that $A=A^{\natural}$. On the other hand, it is well known that there is a matrix $B$ such that $JA=B^{\ast}B$, that is,
$A=JB^{\ast}J^{2}B=B^{\natural}JB$. The other implication is proved similarly.
\end{proof}

An interesting and simple class of maps related with a fundamental symmetry $J$ is one constituted by linear $J$-positive functionals.

\begin{remark}
Let $\Theta$ be a linear $J$-positive functional on $M_{n}(\mathbb{C})$, that is, $\Theta(A)\geq 0$ for all matrix $A$ which is $J$-positive.
Then, one can see that $\Theta$ is necessarily bounded, even more $\|\Theta\|=\Theta(J)$. This result is a consequence of the fact that
$\widehat{\Theta}(A)=\Theta(JA)$ is a positive linear functional and that $X\longrightarrow JX$ is a bijective map. Moreover, $\Theta(A^{\natural})=\overline{\Theta(A)}$ for any $A\in M_{n}(\mathbb{C})$, it implies that if $A$ is a $J$-selfadjoint matrix then
$\Theta(A)$ is real. In this case, $\langle A,B \rangle_{\Theta}=\Theta (AJB^{\natural})$ constitutes a pre-inner product on $M_{n}(\mathbb{C})$
(really an indefinite metric). Observe that from (\ref{s1 3}) it follows that, in particular, $\Theta(A)=Tr\,AJ=Tr\,JA$ is a linear $J$-positive
functional.
\end{remark}

\subsection{The notion of $J$-state in $M_{n}(\mathbb{C})$ provided with an indefinite metric}

Now, we present a notion of quantum state in spaces with an indefinite metric and immediately after, we study some of its properties. In this point,
we recall that a usual quantum state of $\mathbb{C}^{n}$ is a matrix $S$ which is positive and moreover it satisfies that $Tr\,S=1$. The geometry of
the space of all quantum states can be found in \cite{beng}.

We directly propose

\begin{definition}A matrix $B$ will be called a quantum $J$-state once checked that $JB$ is a quantum state.
\end{definition}

We have

\begin{lemma}The matrix $B$ is a quantum $J$-state if and only if $B$ is a $J$-positive matrix and $Tr\,BJ=1$.
\end{lemma}
\begin{proof}Suppose that $B$ is a quantum $J$-state, then $A=JB$ is a quantum state. Hence, $A$ is a positive matrix and $Tr\,A=1$.
First, it implies that $B=JA$ is a $J$-positive matrix and second $Tr\,BJ=Tr\,JAJ=Tr\,A=1$. Conversely, assume that $B$ is $J$-positive
and $Tr\,BJ=1$ then $A=JB$ is positive and also $Tr\,A=Tr\,JAJ=Tr\,BJ=1$, that is, $A$ is a quantum state.
\end{proof}

\begin{example}\label{olv 2 subsec 1}
Observe that the simplest quantum $J$-state is $\Pi =\langle \cdot, e\rangle Je=Je\otimes e$ where $\|e\|=1$. It is called
a pure quantum $J$-state. One can see that $\Pi^{\natural}=\Pi$ and $\Pi J\Pi=\Pi$. The set of all pure quantum $J$-states is denoted by
$\mathfrak{P}_{J}(\mathbb{C})$. A quantum $J$-state that is not pure is called mixed quantum $J$-state.
\end{example}

\begin{theorem}Let $\mathfrak{S}_{J}(\mathbb{C}^{n})$ be the set of all quantum $J$-states. Then, $\mathfrak{S}_{J}(\mathbb{C}^{n})$ is a convex set.
\end{theorem}
\begin{proof}Suppose that $B_{1}, B_{2}\in \mathfrak{S}_{J}(\mathbb{C}^{n})$ and $\beta\in (0,1)$.
From the previous lemma, we must see that $B=\beta B_{1}+(1-\beta)B_{2}$ is $J$-positive and also $Tr\,BJ=1$. Note that both $JB_{1}$ and
$JB_{2}$ are positive matrices, it shows that $JB$ is a positive matrix, in other words, $B$ is a $J$-positive. On the other hand, $Tr\,BJ=
\beta\, Tr\,(B_{1}J)+(1-\beta)\,Tr\,(B_{2}J)=1$.
\end{proof}

\begin{proposition}\label{las niñas 1}A quantum $J$-state is pure, if and only if it is not a convex combination of elements
of $\mathfrak{S}_{J}(\mathbb{C}^{n})$.
\end{proposition}
\begin{proof}
A pure quantum $J$-states $\Pi$ can not be a convex combination of elements in $\mathfrak{S}_{J}(\mathbb{C}^{n})$, otherwise,
$J\Pi$ should be a convex combination of two ordinary quantum states which is impossible. Conversely, if a quantum $J$-states $A$ is
not a convex combination of elements of $\mathfrak{S}_{J}(\mathbb{C}^{n})$, then, it shall be of the form
$\langle \cdot, e \rangle _{\mathbb{C}^{n}}Je$, because in the opposite case
\begin{equation*}
A=\sum_{i=1}^{l}\lambda_{i}\langle \cdot, e_{i}\rangle_{\mathbb{C}^{n}} \,Je_{i}=\sum_{i=1}^{l}\lambda_{i}\,(Je_{i})\otimes e_{i},
\end{equation*}
with $2\leq l$, $\sum \lambda_{i}=1$ and $\{e_{i}\}$ is an orthonormal set. Hence,
\begin{equation*}
A=\lambda_{1}\langle \cdot, e_{1}\rangle_{\mathbb{C}^{n}} \,Je_{1}+(1-\lambda_{1})\sum_{i=2}^{l}\frac{\lambda_{i}}{(1-\lambda_{1})}\langle \cdot, e_{i}\rangle_{\mathbb{C}^{n}} \,Je_{i},
\end{equation*}
which is a contradiction. It shows that the pure quantum $J$-states are the extreme points of convex set $\mathfrak{S}_{J}(\mathbb{C}^{n})$.
\end{proof}

Let $M\in M_{n}(\mathbb{C})$ be a fixed matrix. Then, there exist two orthonormal systems $\{f_{i}\}$ and $\{g_{i}\}$ such that
\begin{equation*}
Mx=\sum_{i=1}^{\nu(M)}s_{i}(M)\langle x, f_{i}\rangle_{\mathbb{C}^{n}} g_{i},\,\,\,\,\,\,\,\,\,\forall x\in \mathbb{C}^{n},
\end{equation*}
where the $s_{i}$ for $i=1,\cdots,\nu(M)$ are the singular values of $M$, that is, the nonzero eigenvalues of $(M^{\ast}M)^{\frac{1}{2}}$.
This is the so-called singular value decomposition of $M$. In the case, when $M^{\ast}=M$ then the $s_{i}(M)$ for $i=1,\cdots,\nu(M)$ are the
eigenvalues of $M$ and $f_{i}=g_{i}$ for all $i=1,\cdots,\nu(M)$. Note that (\ref{s1 3}) constitutes the Schmidt decomposition for a
$J$-positive matrix $A$. Similarly, from (\ref{s1 3}) it follows that $AJ$ is a positive operator because
\begin{equation}
AJ=\sum_{i=1}^{n}\lambda_{i}\langle \cdot, Je_{i}\rangle_{\mathbb{C}^{n}} \,Je_{i}=\sum_{i=1}^{n}\lambda_{i}\,(Je_{i})\otimes Je_{i},
\end{equation}
moreover, $Tr\,AJ=Tr\,JA=\sum_{1}^{n}\lambda_{i}$. Finally, observe that $\langle Je_{i}, Je_{j}\rangle_{\mathbb{C}^{n}}=\delta_{ij}$.

For $1\leq p < \infty$ one defines the following norms
\begin{equation*}
\|M\|_{p}=\left (\sum_{i=1}^{\nu(M)}(s_{i}(M))^{p} \right )^{\frac{1}{p}},
\end{equation*}
and $\|M\|_{\infty}=max_{i}\,s_{i}(M)=\|M\|$. Observe that if $A$ is positive then $\|A\|_{1}=\sum_{i=1}^{\nu(M)}s_{i}(M)=Tr\,A$.
Denote by $\mathcal{S}_{p}$ the Banach space $(M_{n}(\mathbb{C}),\|\cdot\|_{p})$. From now on, the norm $\|\cdot\|_{1}$ is called
the trace-norm. On the other hand, as it was seen before if $A$ is a $J$-positive matrix, then $\|JA\|_{1}=\|AJ\|_{1}$.

Next, we recall some facts related with these spaces $\mathcal{S}_{p}$:
\begin{itemize}
  \item $\|ASB\|_{1}\leq \|A\|\,\|S\|_{1}\,\|B\|,\,\,\,\,\,\,\,\,\,\,\,\,\forall A,S,B\in M_{n},$
  \item For any fixed matrix $T$, the function $F_{T}(\cdot):\,S\longrightarrow F_{T}(S)=Tr(ST)$ defines a continuous linear functional and
  \begin{equation}\label{s1 5}
  |F_{T}(S)|\leq \|S\|_{p}\|T\|_{q},
  \end{equation}
  where $\frac{1}{p}+\frac{1}{q}=1$. In the case $p=1$, this reduces to $|F_{T}(S)|\leq \|S\|_{1}\|T\|$.
\end{itemize}

Observe that if the matrix $T$ is a quantum effect which means that $0\leq T \leq I_{n}$ then $0\leq \|T\|\leq 1$ and so $0\leq F_{T}(S)\leq 1$ for all quantum state $S$. Hence, the value $P_{S}(T)=F_{T}(S)=Tr\, ST$ can be considered as the probability that the effect $T$ emerges in the quantum state $S$, giving rise to a new quantum state.

\begin{definition}A \textbf{quantum $J$-effect} is a matrix $E$ such that $JE$ is a usual quantum effect in the explained above sense.
\end{definition}

We shall denote by $\mathfrak{E}_{J}(\mathbb{C}^{n})$ the set of all $J$-effects which turns out to be a convex set. In fact, if $E_{1}$ and
$E_{2}$ are $J$-effects then $0\leq JE_{1}\leq I_{n}$ and $0\leq JE_{1}\leq I_{n}$ hence for all $\beta\in (0,1)$ we have $0\leq \beta JE_{1}+
(1-\beta)JE_{2}$ and moreover $\beta \langle JE_{1}x,x\rangle_{\mathbb{C}^{n}}+ (1-\beta)\langle JE_{2}x,x\rangle_{\mathbb{C}^{n}}\leq \|x\|^{2}$
for all $x\in\mathbb{C}^{n}$ which means that $0\leq \beta JE_{1}+(1-\beta)JE_{2}\leq I_{n}$. Thus, $\beta E_{1}+(1-\beta)E_{2}$ is a quantum
$J$-effect.

It is useful to introduce $J$-state automorphisms on $\mathfrak{S}_{J}(\mathbb{C}^{n})$.

\begin{definition}A function $\mathfrak{s}:\,\mathfrak{S}_{J}(\mathbb{C}^{n})\longrightarrow \mathfrak{S}_{J}(\mathbb{C}^{n})$ is a $J$-state automorphism if
\begin{enumerate}
  \item The function $\mathfrak{s}$ is a bijection,
  \item $\mathfrak{s}(\beta S_{1}+(1-\beta)S_{2})=\beta\mathfrak{s}(S_{1})+(1-\beta)\mathfrak{s}(S_{2})$ for all $S_{1},S_{2}\in \mathfrak{S}(\mathbb{C}^{n})$ and all $\beta\in (0,1)$.
\end{enumerate}

The case in which $J=I_{n}$, that is, $J$ is the identity matrix can be consulted in \cite{cass}.
\end{definition}

It is clear that the set $Aut_{J}=Aut_{J}(\mathfrak{S}(\mathbb{C}^{n}))$ of all $J$-state automorphisms is a group with respect to the composition of
functions. We recall that $\mathfrak{P}_{J}(\mathbb{C}^{n})=\{\langle \cdot, e \rangle_{\mathbb{C}^{n}}Je\,|\, \|e\|_{\mathbb{C}^{n}}=1\}$ is the set of all pure quantum $J$-states. \\

We already know that $S\in\mathfrak{S}_{J}(\mathbb{C}^{n})$, if and only if $S$ is a $J$-positive matrix and $Tr\, SJ=1$.

\begin{remark}\label{rauli}
From lemma \ref{proposicion1 s1}, we know that if a matrix $A$ is $J$-positive then it is $J$-selfadjoint, that is,
$A^{\natural}=A$. Hence, being $B=C-D$ where $C$ and $D$ are $J$-positive, we have $B^{\natural}=B$. In fact, $B^{\natural}=C^{\natural}-D^{\natural}
=C-D$.
\end{remark}

\begin{proposition}Let $\mathfrak{s}\in Aut_{J}$, then
\begin{itemize}
  \item $\mathfrak{s}$ is the restriction of a unique linear operator $\widetilde{\mathfrak{s}}$
on the real vector space $M_{n}(\mathbb{C})^{sa}(J)=\{A|\,A^{\natural}=A\}$ such that $Tr\,(\widetilde{\mathfrak{s}}(T)J)=Tr\,TJ$ for all
$T\in M_{n}(\mathbb{C})^{sa}(J)$. Moreover, $\widetilde{\mathfrak{s}}$ is a bijection of $M_{n}(\mathbb{C})^{sa}(J)$ on $M_{n}(\mathbb{C})^{sa}(J)$.
  \item moreover, $\mathfrak{s}(\mathfrak{P}_{J})\subset \mathfrak{P}_{J}$.
\end{itemize}
\end{proposition}
\begin{proof}First, we will extend $\mathfrak{s}$ to the set $M^{+}_{n}(\mathbb{C})(J)$ of all $J$-positive matrices. We put
\begin{equation*}
\widetilde{\mathfrak{s}}(T)=\|TJ\|_{1}\,\mathfrak{s}\left(\frac{T}{\|TJ\|_{1}} \right ),\,\,\,\,\,\,\,\forall\, T\neq O_{n}\in M^{+}_{n}(\mathbb{C})(J),
\end{equation*}
and $\widetilde{\mathfrak{s}}(O_{n})=O_{n}$. It is convenient to indicate that the image of a $J$-positive matrix when applying $\widetilde{\mathfrak{s}}$ is, by definition, a $J$-positive matrix.

Now, when $0\leq \lambda$ and $T$ is a $J$-positive matrix, then $\lambda T$ is also $J$-positive. In this case, we obtain $\widetilde{\mathfrak{s}}(\lambda T)=\lambda \widetilde{\mathfrak{s}}( T)$ (the positive homogeneity of
$\widetilde{\mathfrak{s}}$). Indeed,
$$\widetilde{\mathfrak{s}}(\lambda T)=\|\lambda TJ\|_{1}\mathfrak{s}\left (\frac{\lambda T}{\|\lambda TJ\|_{1}} \right )=\lambda
\widetilde{\mathfrak{s}}\left (\frac{T}{\|TJ\|_{1}} \right ).$$

Suppose now that $T_{1}$ and $T_{2}$ are $J$-positive, then we can write the sum $T_{1}+T_{2}$ in the following form
\begin{equation*}
T_{1}+T_{2}=(\|T_{1}J\|_{1}+\|T_{2}J\|_{1})\left (\frac{\|T_{1}J\|_{1}}{(\|T_{1}J\|_{1}+\|T_{2}J\|_{1})}\frac{T_{1}}{\|T_{1}J\|_{1}}+\frac{\|T_{2}J\|_{1}}{(\|T_{1}J\|_{1}+\|T_{2}J\|_{1})}\frac{T_{2}}
{\|T_{2}J\|_{1}} \right ).
\end{equation*}

Let us pay attention to the following details of the previous equality
\begin{itemize}
  \item Observe that $\widehat{T}=\left (\frac{\|T_{1}J\|_{1}}{(\|T_{1}J\|_{1}+\|T_{2}J\|_{1})}\frac{T_{1}}{\|T_{1}J\|_{1}}+\frac{\|T_{2}J\|_{1}}{(\|T_{1}J\|_{1}+\|T_{2}J\|_{1})}\frac{T_{2}}
{\|T_{2}J\|_{1}} \right )\in \mathfrak{S}_{J}(\mathbb{C}^{n})$, because $\widehat{T}$ is a convex combination of two matrices of $\mathfrak{S}_{J}(\mathbb{C}^{n})$.
  \item $T_{1}+T_{2}=\lambda \widehat{T}$ where $\widehat{T}$ is $J$-positive and $0\leq \lambda=(\|T_{1}J\|_{1}+\|T_{2}J\|_{1})$.
\end{itemize}

Hence, the property $2.$ of $\mathfrak{s}$ and the positive homogeneity of $\widetilde{\mathfrak{s}}$ imply that $$\widetilde{\mathfrak{s}}(T_{1}+T_{2})=\widetilde{\mathfrak{s}}(T_{1})+\widetilde{\mathfrak{s}}(T_{2}). $$

Let us extend $\mathfrak{s}$ to $M_{n}(\mathbb{C})^{sa}(J)$. For this purpose, consider a $T\in M_{n}(\mathbb{C})^{sa}(J)$ arbitrary,
then $T^{\natural}=T$, that is, $JT^{\ast}J=T$ which implies that $JT=(JT)^{\ast}$. Hence, $JT=A_{+}-A_{-}$ where $A_{+}$ and $A_{-}$ are usual
positive matrices. It shows that $T=T_{+}-T_{-}$ where both matrices $T_{+}$ and $T_{-}$ are $J$-positive. Then, one defines $\widetilde{\mathfrak{s}}(T)=
\widetilde{\mathfrak{s}}(T_{+})-\widetilde{\mathfrak{s}}(T_{-})$. Now, taking into account that $\widetilde{\mathfrak{s}}(T_{+})$ and $\widetilde{\mathfrak{s}}(T_{-})$ are $J$-positive matrices from remark \ref{rauli}, it follows that $\widetilde{\mathfrak{s}}(T)\in M_{n}(\mathbb{C})^{sa}(J)$.

It is not hard to see that $\widetilde{\mathfrak{s}}$ defined in this form is linear.
Indeed, let $T=\kappa_{1}T_{1}+\kappa_{2}T_{2}$, where $k_{1},\kappa_{2}\in \mathbb{R}$ and $T_{1}, T_{2}\in M_{n}(\mathbb{C})^{sa}(J)$, then
\begin{equation*}
T=\kappa_{1}T_{1}+\kappa_{2}T_{2}=\kappa_{1}[(T_{1})_{+}-(T_{1})_{-}]+\kappa_{2}[(T_{2})_{+}-(T_{2})_{-}],
\end{equation*}
where each of the matrices $(T_{1})_{\pm}, (T_{2})_{\pm}$ is $J$-positive. On the other hand, without loss of generality, we can assume that
$0\leq \kappa_{1},\kappa_{2}$. In fact, otherwise, if for instance $\kappa_{i}< 0$ for some $i$ we have
\begin{equation*}
\kappa_{i}[(T_{i})_{+}-(T_{i})_{-}]=\,-\kappa_{i}[(T_{i})_{-}-(T_{i})_{+}]=\widetilde{\kappa}[(H_{i})_{+}-(H_{i})_{-}],
\end{equation*}
where $0< \widetilde{\kappa}$ and the $(H_{i})_{\pm}$ are $J$-positive matrices. Thus,
\begin{equation*}
T=[\kappa_{1}(T_{1})_{+}+\kappa_{2}(T_{2})_{+}]-[\kappa_{1}(T_{1})_{-}+\kappa_{2}(T_{2})_{-}],
\end{equation*}
note that each $\kappa_{i}(T_{i})_{\pm}$ is a $J$-positive matrix for $i=1,2$. Taking into account the way in which $\widetilde{\mathfrak{s}}$
has been defined, we obtain
\begin{align*}
\widetilde{\mathfrak{s}}(\kappa_{1}T_{1}+\kappa_{2}T_{2})&=\widetilde{\mathfrak{s}}(\kappa_{1}(T_{1})_{+}+\kappa_{2}(T_{2})_{+})-
\widetilde{\mathfrak{s}}
(\kappa_{1}(T_{1})_{-}+\kappa_{2}(T_{2})_{-}) \\
&=\left (\widetilde{\mathfrak{s}}(\kappa_{1}(T_{1})_{+})+\widetilde{\mathfrak{s}}(\kappa_{2}(T_{2})_{+}) \right )-
\left (\widetilde{\mathfrak{s}}(\kappa_{1}(T_{1})_{-})+\widetilde{\mathfrak{s}}(\kappa_{2}(T_{2})_{-}) \right ) \\
&=\left (\kappa_{1}\widetilde{\mathfrak{s}}((T_{1})_{+})+\kappa_{2}\widetilde{\mathfrak{s}}((T_{2})_{+}) \right )-
\left (\kappa_{1}\widetilde{\mathfrak{s}}((T_{1})_{-})+\widetilde{\mathfrak{s}}(\kappa_{2}(T_{2})_{-}) \right ) \\
&=\kappa_{1}\left(\widetilde{\mathfrak{s}}((T_{1})_{+})-\widetilde{\mathfrak{s}}((T_{1})_{-}) \right )+
\kappa_{2}\left(\widetilde{\mathfrak{s}}((T_{2})_{+})-\widetilde{\mathfrak{s}}((T_{2})_{-}) \right ) \\
&=\kappa_{1}\left(\widetilde{\mathfrak{s}}((T_{1})_{+}-(T_{1})_{-}) \right )+\kappa_{2}\left(\widetilde{\mathfrak{s}}((T_{2})_{+}-(T_{2})_{-}) \right ) \\
&=\kappa_{1}\widetilde{\mathfrak{s}}(T_{1})+\kappa_{2}\widetilde{\mathfrak{s}}(T_{2}).
\end{align*}

On the other hand, suppose that $T=T_{1}-T_{2}$ with both $T_{1}$, $T_{2}$ $J$-positive. Then $T_{+}+T_{2}=T_{1}+T_{-}$ from which follows that
$\widetilde{\mathfrak{s}}$ is well defined (because $\widetilde{\mathfrak{s}}(T_{+})-\widetilde{\mathfrak{s}}(T_{-})=\widetilde{\mathfrak{s}}(T_{1})-\widetilde{\mathfrak{s}}(T_{2})$). A main fact
in our construction is that $\widetilde{\mathfrak{s}}$ maps $M^{+}_{n}(\mathbb{C})(J)$ into $M^{+}_{n}(\mathbb{C})(J)$.

Now, let $T$ be an arbitrary $J$-selfadjoint matrix, that is, $T^{\natural}=T$. Then $T=T_{+}-T_{-}$ so
\begin{equation}\label{s1 6}
\widetilde{\mathfrak{s}}(T)=\|T_{+}J\|_{1}\mathfrak{s}\left (\frac{T_{+}}{\|T_{+}J\|_{1}} \right)-\|T_{-}J\|_{1}\mathfrak{s}\left (\frac{T_{-}}{\|T_{-}J\|_{1}} \right),
\end{equation}
observe that $\frac{T_{+}}{\|T_{+}J\|_{1}}$ and $\frac{T_{-}}{\|T_{-}J\|_{1}}$ are quantum $J$-states, hence $\mathfrak{s}(\frac{T_{+}}{\|T_{+}J\|_{1}})$
and $\mathfrak{s}(\frac{T_{-}}{\|T_{-}J\|_{1}})$ are quantum $J$-states and
\begin{equation}\label{s1 7}
Tr\,\left (\mathfrak{s}\left(\frac{T_{+}}{\|T_{+}J\|_{1}}\right)J \right )=1=Tr\,\left (\mathfrak{s}\left(\frac{T_{-}}{\|T_{-}J\|_{1}}\right)J \right ),
\end{equation}
therefore, combining (\ref{s1 6}) and (\ref{s1 7}), we find that
\begin{equation*}
Tr\,(\widetilde{\mathfrak{s}}(T)J)=\|T_{+}J\|_{1}-\|T_{-}J\|_{1}=Tr\,(T_{+}J)-Tr\,(T_{-}J)=Tr\,TJ.
\end{equation*}

Suppose now that $\widehat{\mathfrak{s}}$ is another linear operator which extends $\mathfrak{s}$ such that $\widehat{\mathfrak{s}}(M^{+}_{n}(\mathbb{C})(J)) \subset M^{+}_{n}(\mathbb{C})(J)$. Then, for any Matrix $T$
for which $T^{\natural}=T$, we obtain (using the linearity of $\widehat{\mathfrak{s}}$)
\begin{align*}
\widehat{\mathfrak{s}}(T)&=\widehat{\mathfrak{s}}(T_{+}-T_{-})=\widehat{\mathfrak{s}}(T_{+})-\widehat{\mathfrak{s}}(T_{-})=
\|T_{+}J\|_{1}\widehat{\mathfrak{s}}\left (\frac{T_{+}}{\|T_{+}J\|_{1}}\right )-\|T_{-}J\|_{1}\widehat{\mathfrak{s}}\left (\frac{T_{-}}
{\|T_{-}J\|_{1}} \right) \\
&=\|T_{+}J\|_{1}\mathfrak{s}\left (\frac{T_{+}}{\|T_{+}J\|_{1}}\right )-\|T_{-}J\|_{1}\mathfrak{s}\left (\frac{T_{-}}
{\|T_{-}J\|_{1}} \right)=\widetilde{\mathfrak{s}}(T).
\end{align*}

It shows that $\widetilde{\mathfrak{s}}$ is unique. Notice that if $T_{1}\neq O_{n}$ and $T_{2}\neq O_{n}$ are $J$-positive and $\widetilde{\mathfrak{s}}(T_{1})=\widetilde{\mathfrak{s}}(T_{2})$, then it implies
that $T_{1}=T_{2}$. In fact, if we suppose that
\begin{equation*}
\widetilde{\mathfrak{s}}(T_{1})=\|T_{1}J\|_{1}\mathfrak{s}\left (\frac{T_{1}}{\|T_{1}J\|_{1}} \right )=\|T_{2}J\|_{1}\mathfrak{s}\left (\frac{T_{2}}{\|T_{2}J\|_{1}} \right )=\widetilde{\mathfrak{s}}(T_{2}),
\end{equation*}
then
\begin{equation*}
\|T_{1}J\|_{1}=Tr\,(\widetilde{\mathfrak{s}}(T_{1})J)=Tr\,(\widetilde{\mathfrak{s}}(T_{2})J)=\|T_{2}J\|_{1},
\end{equation*}
and since $\mathfrak{s}$ is a bijection, it follows that $T_{1}=T_{2}$. Next, suppose that $T_{1},T_{2}\in M_{n}(\mathbb{C})^{sa}(J)$
such that $\widetilde{\mathfrak{s}}(T_{1})=\widetilde{\mathfrak{s}}(T_{2})$. Then,
\begin{equation*}
\widetilde{\mathfrak{s}}(T_{1})=\widetilde{\mathfrak{s}}(T^{1}_{+}-T^{1}_{-})=\widetilde{\mathfrak{s}}(T^{1}_{+})-\widetilde{\mathfrak{s}}(T^{1}_{-})
=\widetilde{\mathfrak{s}}(T^{2}_{+})-\widetilde{\mathfrak{s}}(T^{2}_{-})=\widetilde{\mathfrak{s}}(T^{2}_{+}-T^{2}_{-})=\widetilde{\mathfrak{s}}(T_{2}),
\end{equation*}
it shows that
$$\widetilde{\mathfrak{s}}(T^{1}_{+}+T^{2}_{-})=\widetilde{\mathfrak{s}}(T^{1}_{+})+\widetilde{\mathfrak{s}}(T^{2}_{-})
=\widetilde{\mathfrak{s}}(T^{2}_{+})+\widetilde{\mathfrak{s}}(T^{1}_{-})=\widetilde{\mathfrak{s}}(T^{2}_{+}+ T^{1}_{-}),$$
and from this, we obtain $T^{1}_{+}+T^{2}_{-}=T^{2}_{+}+ T^{1}_{-}$ because both sides are $J$-positive matrices, thus $T_{1}=T_{2}$. This tells
us that $\widetilde{\mathfrak{s}}:\,M_{n}(\mathbb{C})^{sa}(J)\longrightarrow M_{n}(\mathbb{C})^{sa}(J)$ is an injective map. Now , we shall show
that it is also surjective. Indeed, let $T\in M_{n}(\mathbb{C})^{sa}(J)$, that is, $T^{\natural}=T$ then $T=T_{+}-T_{-}$ and
\begin{equation*}
T=\|T_{+}J\|_{1}\left (\frac{T_{+}}{\|T_{+}J\|_{1}} \right )-\|T_{-}J\|_{1}\left (\frac{T_{-}}{\|T_{-}J\|_{1}} \right )=
\|T_{+}J\|_{1}\mathfrak{s}(S_{1})-\|T_{-}J\|_{1}\mathfrak{s}(S_{2}),
\end{equation*}
for some $S_{1}, S_{2}\in \mathfrak{S}(\mathbb{C}^{n})$, thus using the linearity of $\widetilde{\mathfrak{s}}$ we have
\begin{equation*}
T=\|T_{+}J\|_{1}\widetilde{\mathfrak{s}}(S_{1})-\|T_{-}J\|_{1}\widetilde{\mathfrak{s}}(S_{2})=\widetilde{\mathfrak{s}}(\|T_{+}J\|_{1}S_{1}
-\|T_{-}J\|_{1}S_{2}).
\end{equation*}

Now, we turn to prove that $\mathfrak{s}(\mathfrak{P}_{J})\subset \mathfrak{P}_{J}$. It is easy to see that $\mathfrak{s}^{-1}(\beta T_{1}+(1-\beta)T_{2})
=\beta\mathfrak{s}^{-1}(T_{1})+(1-\beta)\mathfrak{s}^{-1}(T_{2})$ for all $T_{1}, T_{2}\in \mathfrak{S}_{J}(\mathbb{C}^{n})$ and $\beta \in (0,1)$.
Suppose that $\Pi\in \mathfrak{P}_{J}$ and $\mathfrak{s}(\Pi)=\beta T_{1}+(1-\beta)T_{2}$ then $\Pi=\beta \mathfrak{s}^{-1}(T_{1})+(1-\beta)\mathfrak{s}^{-1} (T_{2})$ so from remark \ref{las niñas 1}, it follows that $\Pi=\mathfrak{s}^{-1}(T_{1})=\mathfrak{s}^{-1}(T_{2})$, that is $\mathfrak{s}(\Pi)=T_{1}=T_{2}$, therefore making use of remark
\ref{las niñas 1} again, we may simply obtain $\mathfrak{s}(\Pi)\in \mathfrak{P}_{J}$. It makes the proof to be concluded\,.
\end{proof}

\begin{remark}Observe that if $\mathfrak{s} \in Aut_{J}$ and $\mathfrak{s}(\Pi)=\Pi$ for all $\Pi\in \mathfrak{P}_{J}$ then $\mathfrak{s}$
is the identity.
\end{remark}

Next, we will construct a concrete J-state automorphism of $\mathfrak{S}_{J}(\mathbb{C}^{n})$.

\begin{example}Define $\mathfrak{s}^{J}_{V}(A)=V^{\natural}AV$ such that $V^{\natural}JV=J=VJV^{\natural}$ (in short, this means that
$V^{\ast}V=VV^{\ast}=I_{n}$)
then $\mathfrak{s}^{J}_{V}\in Aut_{J}$. In fact, if $A\in \mathfrak{S}_{J}(\mathbb{C}^{n})$ we have $[\mathfrak{s}^{J}_{V}(A)x,x]=[V^{\natural}AVx,x]=[AVx,Vx]\geq 0$ for all $x\in \mathbb{C}^{n}$ because $A$ is a $J$-positive matrix.
It shows that $\mathfrak{s}^{J}_{V}(A)$ is a $J$-positive matrix. On the other hand,
\begin{equation*}
Tr\,\mathfrak{s}^{J}_{V}(A)J=Tr\,V^{\natural}AVJ=Tr\,V^{\natural}(AJ)(JVJ)=Tr\,AJ=1,
\end{equation*}
here, we have used that $V^{\natural}(JVJ)=(JVJ)V^{\natural}=J^{2}=I_{n}$. From this follows that $\mathfrak{s}^{J}_{V}(A)\in \mathfrak{S}_{J}(\mathbb{C}^{n})$ and hence $\mathfrak{s}^{J}_{V}(\mathfrak{S}_{J}(\mathbb{C}^{n}))\subset \mathfrak{S}_{J}(\mathbb{C}^{n})$.
Let us assume that $\mathfrak{s}^{J}_{V}(A_{1})=\mathfrak{s}^{J}_{V}(A_{2})$ where $A_{1}, A_{2}\in \mathfrak{S}_{J}(\mathbb{C}^{n})$, then since
$V^{\natural}JV=J=VJV^{\natural}$, it is easy to conclude that $A_{1}=A_{2}$. In other words, $\mathfrak{s}^{J}_{V}$ is injective. Let
$B\in \mathfrak{S}_{J}(\mathbb{C}^{n})$ be arbitrary and define $A=(JVJ)B(JV^{\natural}J)$. Then, one can see that $A\in\mathfrak{S}_{J}(\mathbb{C}^{n})$ and $\mathfrak{s}^{J}_{V}(A)=B$. In the next section, this type of maps will be studied in detail.
\end{example}

This last example shows that the notion of quantum $J$-state as it was introduced in this section is basically related to the usual group
of unitary matrices. Hence, it suggests to call to this as a \textbf{quantum $J$-state of unitary origin}.

\begin{remark}
There is another possible geometrical notion of quantum state in a space with an indefinite metric which is related
with a $J$-unitary matrix.

\begin{definition}We say that a Matrix $A$ is a \textbf{quantum $J$-state of $J$-unitary origin} if $A$ is $J$-positive and $Tr\,A=1$.
\end{definition}

Observe that if $A$ is a quantum $J$-state of $J$-unitary origin, then the same happens with $V^{\natural}AV$ whenever $V^{\natural}V=VV^{\natural}=I_{n}$.
On the other hand, the set of all quantum $J$-state of $J$-unitary origin is a convex set. A pure quantum $J$-state of $J$-unitary origin is
one of the following form $\Gamma=[\cdot, e]e=\langle J\cdot, e\rangle_{\mathbb{C}^{n}}e=\Pi^{\ast}$
(see example \ref{olv 2 subsec 1}).

We would like to indicate that the concept of quantum $J$-state of $J$-unitary origin leads to a symmetric theory to which was previously developed
with the definition of $J$-state of unitary origin. We show this fact with the following table (for a matrix $M$ the notation $0\leq_{J} M$ means
that $M$ is a $J$-positive matrix),

\bigskip

\noindent\begin{tabular}{|c|c|c|c|}\hline $J$-state of unit. orig. $B$ & then if A=JB & $J$-state of $J$-unit. orig. $B$ & then if $A=JB$ \\
\hline
$0\leq_{J}B$,\,\,\,\,$Tr\,BJ=1$ & $0\leq A$,\,\,\,\,\,$Tr\,A=1$ & $0\leq_{J}B$,\,\,\,\,$Tr\,B=1$ & $0\leq A$,\,\,\,$Tr\,AJ=1$ \\
\hline
\end{tabular}\,.
\end{remark}

\smallskip

\section{Completely $J$-positive type maps on $M_{n}(\mathbb{C})$ and quantum $J$-channel}

The study of positive maps on $C^{\ast}$-algebras began long before the boom of the quantum theories of computation and information;
which were suggested mainly by Paul Benioff, Richard Feynman and Yuri Manin in the $80$s of the last century. These maps were introduced around
$1950$ by R. V. Kadison in the papers \cite{kadison1}, \cite{kadison2}. Later in $1955$ Stinespring introduced completely
positive maps and proved his important dilation theorem \cite{stines}, simultaneously. The relationship between completely positive maps and
the theory of dilation was extensively formalized by Arveson \cite{arveson1}, \cite{arveson2} and \cite{arveson3}.
At that time, the topic was not popular among mathematicians and practically it was barely known to people from other fields, however
remarkable progress was made. The situation changed in the 1990s when the importance of completely positive maps in quantum information
theory was evidenced. It can be said that, currently, the subject is consolidated and the people$'$s interest about it is constantly increasing.

We want to mention that the study of certain types of completely positive maps on groups, using representation theory on Hilbert spaces equipped
with an indefinite metric defined by means of fundamental symmetries, it was first carried out by J. Heo in \cite{heo}. There are excellent texts
on completely positive maps among which we have only selected a few of them: \cite{beng},\cite{paulsen} and \cite{stormer}.

On the other hand, theoretically quantum channels or quantum operators are the fundamental objects through which information is transmitted.
They constitute completely positive maps that preserve the matrix trace. A detailed discussion of the theory of quantum channels in the
finite dimensional case is presented in \cite{wat}.
In the class of quantum channels, we must mention some of them that have particular characteristics and perform important functions within the theory of quantum information. For example, \textbf{random unitary quantum channels} which have the form
\begin{equation*}
\Phi(A)=\sum_{s=1}^{l}\mu_{s}U^{^{\ast}}_{s}AU_{s},
\end{equation*}
where each $U_{s}$ is a unitary matrix and $\mu_{s}$ are positive weights, such that $\sum_{s=1}^{l}\mu_{s}=1$ \cite{collins}.

This class of quantum channels is very important, since the action of such channels can be considered as the random application of one
of the unitary transformations $U_{s}$, with respective probabilities $\mu_{s}$. Moreover, because these have particular properties.
For example, random unitary channels have been used to disprove the additivity of minimum output entropy,
see \cite{hastings}. We refer to some papers, where this kind of quantum channels has been studied, for example, \cite{aubrun},
\cite{hastings} and \cite{hayden}.

Other important quantum channels are the \textbf{quantum Gaussian channels} which have certain behavior with respect to the so-called
characteristic function on trace-class matrices, these have a main role in quantum communication theory because they determine the
attenuation and the noise affecting any electromagnetic signal, in the quantum regime. In this regard, the reader can consult e.g
\cite{de palma} and \cite{de palma1}.

The purpose of this section is to study completely positive maps and quantum channels between indefinite metric spaces. Throughout
this section the fundamental symmetry $J\in M_{n}$ is fixed.

\subsection{Kraus $J$-maps and Completely $J$-positive maps}

In the space $M_{n}(\mathbb{C})$ of $n\times n$ matrices, we consider the following kind of map
\begin{equation}\label{s2 1}
\Phi(A)=\sum_{s=1}^{\nu}V_{s}^{\natural}AV_{s},\,\,\,\,\,\,\,\,\,\forall\,A\in M_{n}(\mathbb{C}),
\end{equation}
here $M^{\natural}$ denotes the $J$-adjoint for an arbitrary matrix $M$ and $(V_{1},\cdots,V_{\nu})\in (M_{n}(\mathbb{C}))^{\nu}$
is a fixed matrix vector. These linear maps are called \textbf{Kraus $J$-maps} by us, and the number $\nu$ is named the Kraus index
for the corresponding $\Phi$. In addition, observe that they transform $J$-positive $n\times n$ matrices into $J$-positive matrices of the
same order. Indeed, from (\ref{s2 1}), it follows that for all $x\in \mathbb{C}^{n}$ and any $A\in M_{n}^{+}(\mathbb{C})(J)$
\begin{equation*}
[\Phi(A)x,x]=\left [\left (\sum_{s=1}^{\nu}V_{s}^{\natural}AV_{s}\right)x,x \right ]=\sum_{s=1}^{\nu}[V_{s}^{\natural}AV_{s}x,x]
=\sum_{s=1}^{\nu}[AV_{s}x,V_{s}x]\geq 0.
\end{equation*}

We say that $\Phi:M_{n}(\mathbb{C})\longrightarrow M_{n}(\mathbb{C})$ linear is \textbf{$J$-positive} if $\Phi \left (M_{n}^{+}(\mathbb{C})(J)\right )
\subset M_{n}^{+}(\mathbb{C})(J)$. Thus, the map $\Phi$ defined by (\ref{s2 1}) is $J$-positive. Let us assume that $\Phi(\cdot)$ is a Kraus
$J$-map then $\Phi_{M}(\cdot)=M^{\natural}\Phi(\cdot)M$ is also a Kraus $J$-map for all matrix $M$ of order $n$, and the Kraus index of $\Phi$
and $\Phi_{M}$ are equal. The simplest Kraus $J$-map is the identity map, that is $\Phi(A)=A$, because $I_{n}^{\natural}=I_{n}$ and so
$\Phi(A)=A=I_{n}^{\natural}AI_{n}$.

The following lemma will be very important for our goals. It shows the general form of a $J$-positive map.

\begin{lemma}\label{lema 2}
Let $\Psi$ be a $J$-positive map on $M_{n}(\mathbb{C})$, then there is $\Phi:\,M_{n}(\mathbb{C})\longrightarrow M_{n}(\mathbb{C})$ positive,
such that, $\Psi(\cdot)=J\Phi(J\,\cdot)$\,.
\end{lemma}
\begin{proof}We directly define $\Phi(\cdot)=J\Psi(J\,\cdot)$ which is evidently a linear map because $\Psi$ is linear by definition. Let $N\in M_{n}^{+}(\mathbb{C})$ be arbitrary then $JN$ is a $J$-positive matrix. Hence, taking into account that $\Psi$ is a $J$-positive map, we shall have $\Psi(JN)=JL$ for some $L\in M_{n}^{+}(\mathbb{C})$. It shows that $\Phi(N)=L$ and so $\Phi$ is a positive map. It is now easy to see that $\Psi(\cdot)=J\Phi(J\cdot)$.
\end{proof}

Let $\Psi^{jp}$ be a $J$-positive map. Then, the positive map which was described in the lemma \ref{lema 2} and that it was put in
correspondence with $\Psi^{jp}$ will be denoted by $\Phi^{p}_{\Psi^{jp}}$. We refer to $\Phi^{p}_{\Psi^{jp}}$ as the positive map
associated to $\Psi^{jp}$. Below, such a correspondence shall also be indicated in the following way $\Psi\,\,
\overrightarrow{jp\longrightarrow p} \,\,\Phi$.

\begin{theorem}\label{teorema1}
Suppose that $\Phi(\cdot)$ is a usual completely positive map on $M_{n}(\mathbb{C})$ then $\Psi(\cdot)=J\Phi(J\cdot)$ is a Kraus
$J$-map on $M_{n}(\mathbb{C})$. Conversely, if $\Psi(\cdot)$ is a Kraus $J$-map on $M_{n}(\mathbb{C})$ then $\Phi(\cdot)=\Psi(J\cdot)J$ is
a completely positive map on $M_{n}(\mathbb{C})$.
\end{theorem}
\begin{proof}Since $\Phi$ is a completely positive map then it can be represented in the following form
\begin{equation}\label{s2 2}
\Phi(A)=\sum_{s=1}^{\nu}V_{s}^{\ast}AV_{s},\,\,\,\,\,\,\,\,\,\forall\,A\in M_{n}(\mathbb{C}),
\end{equation}
for some $\nu\in\mathbb{Z}_{+}$, and some matrices $V_{s}\in M_{n}(\mathbb{C})$ for $s=1,\cdots,\nu$. This result can be found in \cite{wat},
page $82$. Hence, from (\ref{s2 2}) it follows that
\begin{equation*}
\Psi(A)=J\Phi(JA)=\sum_{s=1}^{\nu}JV_{s}^{\ast}JAV_{s}=\sum_{s=1}^{\nu}V_{s}^{\natural}AV_{s},\,\,\,\,\,\,\,\,\,\forall\,A\in M_{n}(\mathbb{C}).
\end{equation*}

Next, we prove the converse. Suppose that $\Psi$ is  a Kraus $J$-map on $M_{n}(\mathbb{C})$, that is
\begin{equation*}
\Psi(A)=\sum_{s=1}^{\nu}V_{s}^{\natural}AV_{s},\,\,\,\,\,\,\,\,\,\forall\,A\in M_{n}(\mathbb{C}),
\end{equation*}
where $\nu\in \mathbb{Z}_{+}$ and the $V_{s}\in M_{n}(\mathbb{C})$ for $s=1,\cdots,\nu$ depend of $\Psi$, then
\begin{equation*}
\Phi(A)=\Psi(JA)J=\sum_{s=1}^{\nu}JV_{s}^{\ast}AV_{s}J=\sum_{s=1}^{\nu}(V_{s}J)^{\ast}A(V_{s}J),\,\,\,\,\,\,\,\,\,\forall\,A\in M_{n}(\mathbb{C}),
\end{equation*}
then, from the result in \cite{wat} page $82$ previously referred, we conclude that $\Phi$ is a completely positive map.
\end{proof}

From now on, each $kn\times kn$ matrix will be written in block form
\begin{equation*}
\mathcal{C}=\left(
              \begin{array}{ccc}
                C_{11} & \cdots & C_{1k} \\
                \vdots & \ddots & \vdots \\
                C_{k1} & \cdots & C_{kk} \\
              \end{array}
            \right),
\end{equation*}
here each block $C_{ij}$, $i,j=1,\cdots,k$ is a complex $n\times n$ matrix and $k=1,2,\cdots$. In particular, for later use, we denote
by $\mathcal{J}_{k}$ the following block $kn\times kn$ diagonal matrix
\begin{equation*}
\mathcal{J}_{k}=\left(
              \begin{array}{cccc}
                J & O_{n} & \cdots & O_{n} \\
                O_{n} & \ddots & \ddots & \vdots \\
                \vdots & \ddots & \ddots & O_{n} \\
                O_{n} & \cdots & O_{n} & J \\
              \end{array}
            \right),
\end{equation*}
where $k=1,2,\cdots$. It is clear that $\mathcal{J}_{k}$ is a fundamental symmetry on $\mathbb{C}^{kn}$ for all $k\geq 1$, that is,
$\mathcal{J}_{k}^{\ast}=\mathcal{J}_{k}$ and $\mathcal{J}_{k}^{2}=I_{kn}$; finally observe that $\mathcal{J}_{1}=J$.
On the other hand, recall that the $kn\times kn$ block matrix $\mathcal{J}_{k}$ induces the following indefinite metric on $\mathbb{C}^{kn}$
$$[x,y]_{\mathcal{J}_{k}}=\langle \mathcal{J}_{k}x,y \rangle_{\mathbb{C}^{kn}}\,,$$
for all $k=1,2,\cdots$.

Let $\Phi:M_{n}(\mathbb{C})\longrightarrow M_{n}(\mathbb{C})$ given. Then, it induces for each $k\in\mathbb{N}$ a map $\Phi^{k}:M_{kn}(\mathbb{C})\longrightarrow M_{kn}(\mathbb{C})$ which is defined in the following form
\begin{equation*}
\Phi^{k}(\mathcal{C})=\left(
              \begin{array}{ccc}
                \Phi(C_{11}) & \cdots & \Phi(C_{1k}) \\
                \vdots & \ddots & \vdots \\
                \Phi(C_{k1}) & \cdots & \Phi(C_{kk}) \\
              \end{array}
            \right),
\end{equation*}
where, clearly $\Phi^{1}=\Phi$.

\begin{definition}We say that $\Phi:M_{n}(\mathbb{C})\longrightarrow M_{n}(\mathbb{C})$ is a \textbf{completely $J$-positive map}
if, for all $k\geq 1$, the map $\Phi^{k}$ is $\mathcal{J}_{k}$-positive. In other words, $\Phi$ is completely $J$-positive if,
for all $k=1,2,\cdots$,
\begin{equation*}
\Phi^{k}(M_{kn}^{+}(\mathbb{C})(\mathcal{J}_{k}))\subset M_{kn}^{+}(\mathbb{C})(\mathcal{J}_{k}).
\end{equation*}
\end{definition}

\begin{lemma}\label{lema1}
Let $\Phi: M_{n}(\mathbb{C})\longrightarrow M_{n}(\mathbb{C})$ be a fixed map and $\Psi(\cdot)=J\Phi(J\cdot)$, then
$\Psi^{k}(\cdot)=\mathcal{J}_{k}\Phi^{k}(\mathcal{J}_{k}\,\cdot)$ for all $k\geq 1$. On the other hand, if we define
$\Theta(\cdot)=\Phi(J\cdot)J$, one has $\Theta^{k}(\cdot)=\Phi^{k}(\mathcal{J}_{k}\,\cdot)\mathcal{J}_{k}$.
\end{lemma}
\begin{proof}If we choose an arbitrary block matrix $\mathcal{C}\in M_{kn}(\mathbb{C})$ then, we obtain
\begin{align*}
\Psi^{k}(\mathcal{C})&=\left(
              \begin{array}{ccc}
                \Psi(C_{11}) & \cdots & \Psi(C_{1k}) \\
                \vdots & \ddots & \vdots \\
                \Psi(C_{k1}) & \cdots & \Psi(C_{kk}) \\
              \end{array}
            \right)=\left(
              \begin{array}{ccc}
                J\Phi(JC_{11}) & \cdots & J\Phi(JC_{1k}) \\
                \vdots & \ddots & \vdots \\
                J\Phi(JC_{k1}) & \cdots & J\Phi(JC_{kk}) \\
              \end{array}
            \right) \\
&=\left(
              \begin{array}{cccc}
                J & O_{n} & \cdots & O_{n} \\
                O_{n} & \ddots & \ddots & \vdots \\
                \vdots & \ddots & \ddots & O_{n} \\
                O_{n} & \cdots & O_{n} & J \\
              \end{array}
            \right)\left(
              \begin{array}{cccc}
                \Phi(JC_{11}) & \cdots & \cdots & \Phi(JC_{1k}) \\
                \vdots & \ddots & \ddots & \vdots \\
                \vdots & \ddots & \ddots & \vdots \\
                \Phi(JC_{k1}) & \cdots & \cdots & \Phi(JC_{kk}) \\
              \end{array}
            \right)=\mathcal{J}_{k}\Phi^{k}(\mathcal{J}_{k}\mathcal{C}).
\end{align*}

The proof that $\Theta^{k}(\cdot)=\Phi^{k}(\mathcal{J}_{k}\,\cdot)\mathcal{J}_{k}$ is very similar; therefore, it will be omitted.
\end{proof} \\

The particular selection of $\Psi$ in the next proposition is justified by virtue of the lemma \ref{lema 2}. We have

\begin{theorem}\label{proposicion 1}
Let $\Phi: M_{n}(\mathbb{C})\longrightarrow M_{n}(\mathbb{C})$ be a completely positive map, that is, for all $k=1,2,\cdots$,
$\Phi^{k}:M_{kn}(\mathbb{C})\longrightarrow M_{kn}(\mathbb{C})$ is a habitual positive map. Then, $\Psi(\cdot)=J\Phi(J\,\cdot)$ is a
completely $J$-positive map.
\end{theorem}
\begin{proof}
Since $\Phi: M_{n}(\mathbb{C})\longrightarrow M_{n}(\mathbb{C})$ is a completely positive map it admits a representation
of the form
\begin{equation*}
\Phi(A)=\sum_{i=1}^{\nu}V_{i}^{\ast}AV_{i},\,\,\,\,\,\,\,\,\,\forall\,A\in M_{n}(\mathbb{C}).
\end{equation*}

We must prove that $\Psi^{k}$ is $\mathcal{J}_{k}$-positive for all $k\in\mathbb{Z}_{+}$. From lemma \ref{lema1}, we have
$\Psi^{k}(\cdot)=\mathcal{J}_{k}\Phi^{k}(\mathcal{J}_{k}\,\cdot)$ for all $k\in \mathbb{Z}_{+}$. Now, let $k$ be fixed but arbitrary and
let $\mathcal{C}$ be any matrix of order $kn$ which is $\mathcal{J}_{k}$-positive, that is, $0\leq[\mathcal{C}x,x]_{\mathcal{J}_{k}}$ for
all $x\in\mathbb{C}^{kn}$. Then, for every $x^{T}=(x_{1},\cdots,x_{k}), y^{T}=(y_{1},\cdots,y_{k})\in \mathbb{C}^{kn}$ we obtain
\begin{align*}
&[\Psi^{k}(\mathcal{C})x,y]_{\mathcal{J}_{k}}=[\mathcal{J}_{k}\Phi^{k}(\mathcal{J}_{k}\mathcal{C})x,y]_{\mathcal{J}_{k}}=
\langle \mathcal{J}_{k}\Phi^{k}(\mathcal{J}_{k}\mathcal{C})x,y \rangle_{\mathbb{C}^{kn}}=\langle \Phi^{k}(\mathcal{J}_{k}\mathcal{C})x,y \rangle_{\mathbb{C}^{kn}}= \\
&\left <\left(
              \begin{array}{ccc}
                \Phi(JC_{11}) & \cdots & \Phi(JC_{1k}) \\
                \vdots & \ddots & \vdots \\
                \Phi(JC_{k1}) & \cdots & \Phi(JC_{kk}) \\
              \end{array}
            \right)x,y \right >
=\left <\left(
              \begin{array}{ccc}
                \sum_{i=1}^{\nu}V_{i}^{\ast}JC_{11}V_{i} & \cdots & \sum_{i=1}^{\nu}V_{i}^{\ast}JC_{1k}V_{i} \\
                \vdots & \ddots & \vdots \\
                \sum_{i=1}^{\nu}V_{i}^{\ast}JC_{k1}V_{i} & \cdots & \sum_{i=1}^{\nu}V_{i}^{\ast}JC_{kk}V_{i} \\
              \end{array}
            \right)x,y \right >  \\
            &=\left <\sum_{i=1}^{\nu}\left(
              \begin{array}{ccc}
                V_{i}^{\ast}JC_{11}V_{i} & \cdots & V_{i}^{\ast}JC_{1k}V_{i} \\
                \vdots & \ddots & \vdots \\
                V_{i}^{\ast}JC_{k1}V_{i} & \cdots & V_{i}^{\ast}JC_{kk}V_{i} \\
              \end{array}
            \right)x,y \right >
            =\sum_{i=1}^{\nu}\left <\left(
              \begin{array}{ccc}
                V_{i}^{\ast}JC_{11}V_{i} & \cdots & V_{i}^{\ast}JC_{1k}V_{i} \\
                \vdots & \ddots & \vdots \\
                V_{i}^{\ast}JC_{k1}V_{i} & \cdots & V_{i}^{\ast}JC_{kk}V_{i} \\
              \end{array}
            \right)x,y \right > \\
            &=\sum_{i=1}^{\nu}\left <\left(
              \begin{array}{ccc}
                JC_{11} & \cdots & JC_{1k} \\
                \vdots & \ddots & \vdots \\
                JC_{k1} & \cdots & JC_{kk} \\
              \end{array}
            \right)\left(
                     \begin{array}{c}
                       V_{i}x_{1} \\
                       \vdots \\
                       V_{i}x_{k} \\
                     \end{array}
                   \right)
            ,\left(
                     \begin{array}{c}
                       V_{i}y_{1} \\
                       \vdots \\
                       V_{i}y_{k} \\
                     \end{array}
                   \right) \right >  \\
                   &=\sum_{i=1}^{\nu}\left <\mathcal{J}_{k}\left(
              \begin{array}{ccc}
                C_{11} & \cdots & C_{1k} \\
                \vdots & \ddots & \vdots \\
                C_{k1} & \cdots & C_{kk} \\
              \end{array}
            \right)\left(
                     \begin{array}{c}
                       V_{i}x_{1} \\
                       \vdots \\
                       V_{i}x_{k} \\
                     \end{array}
                   \right)
            ,\left(
                     \begin{array}{c}
                       V_{i}y_{1} \\
                       \vdots \\
                       V_{i}y_{k} \\
                     \end{array}
                   \right) \right >  \\
                   &=\sum_{i=1}^{\nu}\left [ \left(
              \begin{array}{ccc}
                C_{11} & \cdots & C_{1k} \\
                \vdots & \ddots & \vdots \\
                C_{k1} & \cdots & C_{kk} \\
              \end{array}
            \right)\left(
                     \begin{array}{c}
                       V_{i}x_{1} \\
                       \vdots \\
                       V_{i}x_{k} \\
                     \end{array}
                   \right)
            ,\left(
                     \begin{array}{c}
                       V_{i}y_{1} \\
                       \vdots \\
                       V_{i}y_{k} \\
                     \end{array}
                   \right) \right ]_{\mathcal{J}_{k}},
\end{align*}
now, taking into account that $\mathcal{C}$ is a $\mathcal{J}_{k}$-positive matrix, we conclude that $0\leq[\Psi^{k}(\mathcal{C})z,z]_{\mathcal{J}_{k}}$
for all $z\in\mathcal{C}^{kn}$.
\end{proof}

Next, we summarize some remarks. Observe that from the proof of the previous theorem, it follows that if $\Phi(A)=\sum_{i=1}^{\nu}V_{i}^{\ast}AV_{i}$
for any $A\in M_{n}(\mathbb{C})$, being $\Phi$ completely positive and we define the map $\Psi(\cdot)=J\Phi(J\,\cdot)$, then for all $k\in\mathbb{Z}_{+}$,
one obtains
\begin{equation}\label{s2 3}
\Psi^{k}(\mathcal{C})=\mathcal{J}_{k}\Phi^{k}(\mathcal{J}_{k}\mathcal{C})=\sum_{i=1}^{\nu}\mathcal{V}^{\natural}_{i}\mathcal{C}\mathcal{V}_{i},
\,\,\,\,\,\,\,\,\,\,\,\,\,\forall\,\mathcal{C}\in M_{kn}(\mathbb{C}),
\end{equation}
where for $i=1,\cdots,\nu$
\begin{equation}\label{s2 4}
\mathcal{V}_{i}=\left(
              \begin{array}{cccc}
                V_{i} & O_{n} & \cdots & O_{n} \\
                O_{n} & \ddots & \ddots & \vdots \\
                \vdots & \ddots & \ddots & O_{n} \\
                O_{n} & \cdots & O_{n} & V_{i} \\
              \end{array}
            \right),\,\,\,\,\,\,\,\,\,\,\,\,\,\mathcal{V}^{\natural}_{i}=\left(
              \begin{array}{cccc}
                V^{\natural}_{i} & O_{n} & \cdots & O_{n} \\
                O_{n} & \ddots & \ddots & \vdots \\
                \vdots & \ddots & \ddots & O_{n} \\
                O_{n} & \cdots & O_{n} & V^{\natural}_{i} \\
              \end{array}
            \right),
\end{equation}
are matrices of order $kn$ and for each $i$ the matrix $\mathcal{V}^{\natural}_{i}$ denotes the $\mathcal{J}_{k}$-adjoint of $\mathcal{V}_{i}$ with respect to the indefinite metric $[\cdot,\cdot]_{\mathcal{J}_{k}}=\langle \mathcal{J}_{k} \cdot, \cdot \rangle_{\mathbb{C}^{kn}}$. On the other hand, for all
$k\in \mathbb{Z}_{+}$, (\ref{s2 3}) implies
\begin{equation*}
\Psi^{k}(\mathcal{J}_{k}\mathcal{C})\mathcal{J}_{k}=\sum_{i=1}^{\nu}\mathcal{J}_{k}\mathcal{V}^{\ast}_{i}\mathcal{C}\mathcal{V}_{i}\mathcal{J}_{k}
=\sum_{i=1}^{\nu}\mathcal{D}^{\ast}_{i}\mathcal{C}\mathcal{D}_{i}=\mathcal{J}_{k}\Phi^{k}(\mathcal{C})\mathcal{J}_{k},
\end{equation*}
in which $\mathcal{D}_{i}=\mathcal{V}_{i}\mathcal{J}_{k}$ for $i=1,\cdots,\nu$. It shows that for all $k\in \mathbb{Z}_{+}$
\begin{equation}\label{s2 4}
\Phi^{k}(\cdot)=\mathcal{J}_{k}\Psi^{k}(\mathcal{J}_{k}\,\cdot)\,\,.
\end{equation}

\begin{proposition}\label{cierra1}
Each Kraus $J$-map
\begin{equation*}
\Psi(A)=\sum_{i=1}^{\nu}W_{i}^{\natural}AW_{i},\,\,\,\,\,\,\,\,\,\forall\,A\in M_{n}(\mathbb{C}),
\end{equation*}
is a completely $J$-positive map.
\end{proposition}
\begin{proof}Define $\Phi(A)=J\Psi(JA)$ for all $A\in M_{n}(\mathbb{C})$. It is clear that $\Psi\overrightarrow{jp\longrightarrow p}\,\Phi$
in the sense of lemma \ref{lema 2}. Then,
\begin{equation*}
\Phi(A)=\sum_{i=1}^{\nu}W_{i}^{\ast}AW_{i},\,\,\,\,\,\,\,\,\,\forall\,A\in M_{n}(\mathbb{C}),
\end{equation*}
which implies that $\Phi$ is a completely positive map (see \cite{wat}, page 82). Now, the proposition follows from theorem \ref{proposicion 1}.
\end{proof}

\subsection{Admissible Kraus $J$-positive maps and Quantum $J$-channels}

We recall that an ordinary quantum operator (or ordinary quantum channel) is a completely positive matrix map $\Phi$ which is trace-preserving,
that is, $Tr\, \Phi(A)= Tr\, A$. Clearly, a quantum operator maps states into states, where by a state, we mean a positive matrix whose trace is $1$.
It is well known that $\Phi$ is a quantum channel, if and only if (see \cite{wat} page $89$)
\begin{equation}\label{s2 5}
\Phi(A)=\sum_{i=1}^{\nu}V_{i}^{\ast}AV_{i},\,\,\,\,\,\,\,\,\,\forall\,A\in M_{n}(\mathbb{C}),
\end{equation}
and
\begin{equation}\label{s2 6}
\sum_{i=1}^{\nu}V_{i}V_{i}^{\ast}=I_{n}\,.
\end{equation}

From (\ref{s2 5}) and (\ref{s2 6}) follow that being $\Phi$ a quantum operator and $U\in M_{n}(\mathbb{C})$ a unitary matrix, that is,
$U^{\ast}U=UU^{\ast}=I_{n}$, then $\Theta(\cdot)=U^{\ast}\Phi(\cdot)U$ is a quantum channel.

Suppose that $\Phi$ is a quantum operator, then we already know that $\Psi(\cdot)=J\Phi(J\,\cdot)$ is a completely $J$-positive map, and
from (\ref{s2 5}) we obtain
\begin{equation}\label{s2 7}
\Psi(A)=\sum_{i=1}^{\nu}V_{i}^{\natural}AV_{i},\,\,\,\,\,\,\,\,\,\forall\,A\in M_{n}(\mathbb{C}),
\end{equation}
and (\ref{s2 6}) implies
\begin{equation}\label{s2 8}
\sum_{i=1}^{\nu}V_{i}JV^{\natural}_{i}=J.
\end{equation}

Condition (\ref{s2 6}) is equivalent to (\ref{s2 8}). However, in our exposition, we will maintain (\ref{s2 8}) to achieve a more related
format to the new situation. In other words, we would only like to use the involution $\natural$ in the involved equations.

A map $\Psi$ satisfying (\ref{s2 7}) and (\ref{s2 8}) is called an \textbf{admissible Kraus $J$-positive map}.
Let $U$ be a $J$-unitary $n\times n$
matrix, that is, $UU^{\natural}=U^{\natural}U=I_{n}$ and assume that $\Psi$ is an admissible Kraus $J$-positive map, that is, we can find
$\nu\in\mathbb{Z}_{+}$ and a matrix vector $(V_{1},\cdots,V_{\nu})\in(M_{n}(\mathbb{C}))^{\nu}$ which depend on $\Psi$ such that
(\ref{s2 7})-(\ref{s2 8}) hold then $\Theta(\cdot)=U^{\natural}\Psi(\cdot)U$ is also an admissible Kraus $J$-positive map. To prove this
fact, it is enough to note that
$$(AB)^{\natural}=J(AB)^{\ast}J=JB^{\ast}A^{\ast}J=JB^{\ast}J^{2}A^{\ast}J=B^{\natural}A^{\natural}.$$

\begin{definition}The map $\Psi$, is said to be a \textbf{quantum $J$-operator} (or \textbf{quantum $J$-channel}) if it is a completely $J$-positive
map such that $J\Psi(J\,\cdot)$ is trace-preserving (it implies that $\Psi(J\,\cdot)J$ is also trace-preserving).
\end{definition}

Now, we will discuss the relation between admissible Kraus $J$-positive maps and quantum $J$-channels.

\begin{theorem}\label{teorema 2 s2}
Assume that $\Psi$ is an admissible Kraus $J$-positive map, then it is a quantum $J$-channel.
\end{theorem}
\begin{proof}Suppose that $\Psi$ is an admissible Kraus $J$-positive map, then there is $\nu$ and there exists a matrix vector $(V_{1},\cdots,V_{\nu})\in(M_{n}(\mathbb{C}))^{\nu}$ such that (\ref{s2 8}) holds and $\Psi$ admits the representation (\ref{s2 7}).
Define $\Phi(\cdot)=J\Psi(J\,\cdot)$, so from (\ref{s2 7}) and (\ref{s2 8}), we can recover (\ref{s2 5}) and (\ref{s2 6})
which imply that $\Phi$ is a quantum channel. Hence, theorem \ref{proposicion 1} shows that $\Psi$ is a completely $J$-positive map.
On the other hand, $J\Psi(J\,\cdot)=\Phi(\cdot)$ is trace-preserving.
\end{proof}

\begin{proposition}Assume that $\Psi$ is a completely $J$-positive map and $\Phi$ such that $\Psi\,\, \overrightarrow{jp\longrightarrow p} \,\,\Phi$
then $\Phi$ is a usual completely positive map. Even more, suppose that $\Psi$ is a quantum $J$-operator, then this $\Phi$ will be an ordinary quantum channel.
\end{proposition}

This proposition can be considered a reciprocal one, in regard to proposition \ref{proposicion 1}. \\

\begin{proof}From lemma \ref{lema1}, one knows that $\Phi^{k}(\cdot)=\mathcal{J}_{k}\Psi^{k}(\mathcal{J}_{k}\,\cdot)$ for all $k\in\mathbb{Z}_{+}$
where by hypothesis each $\Psi^{k}$ is $\mathcal{J}_{k}$-positive. It implies that $\Psi^{k}(\mathcal{J}_{k}N)\in M^{+}_{kn}(\mathbb{C})(\mathcal{J}_{k})$
for all $N\in  M^{+}_{kn}(\mathbb{C})$ so $\Psi^{k}(\mathcal{J}_{k}N)=\mathcal{J}_{k}L$ for some $L\in  M^{+}_{kn}(\mathbb{C})$. It follows that
$\Phi^{k}$ is positive on $M_{kn}(\mathbb{C})$. Indeed, $\Phi^{k}(N)=L$, in other words, $\Phi^{k}(M^{+}_{kn}(\mathbb{C}))\subset M^{+}_{kn}(\mathbb{C})$.
On the other hand, being $\Psi$ a quantum $J$-operator, by definition, it is a completely $J$-positive, so from the first part of our proof it follows
that $\Phi$ is a completely positive map. Finally, since $J\Psi(J\,\cdot)$ is preserving trace then $\Phi(\cdot)=J\Psi(J\,\cdot)$ is a quantum operator.
The proposition has been proved.
\end{proof}

\begin{proposition}\label{corolario 1 s2}
Each completely $J$-positive map $\Psi$ is a Kraus $J$-map. Even more, from (\ref{s2 3}) it follows that each $\Psi^{k}$ is a Kraus
$\mathcal{J}_{k}$-map. Finally, every quantum $J$-channel is an admissible Kraus $J$-positive map.
\end{proposition}
\begin{proof}Let $\Phi$ be the completely positive map such that $\Psi\,\, \overrightarrow{jp\longrightarrow p} \,\,\Phi$ (we are taking into account
the previous proposition), then
\begin{equation*}
\Phi(A)=\sum_{i=1}^{v}V_{i}^{\ast} AV_{i},\,\,\,\,\,\,\,\,\,\,\,\forall\,A\in M_{n}(\mathbb{C}),
\end{equation*}
and since $\Phi(\cdot)=J\Psi(J\,\cdot)$, the statement follows. If, on the contrary, we suppose something stronger, let's say, that $\Psi$ is
a quantum $J$-channel, then this $\Phi$ is now a quantum operator, so the matrix vector $(V_{1},\cdots,V_{\nu})$ satisfies the condition
\begin{equation*}
\sum_{i=1}^{\nu}V_{i}V_{i}^{\ast}=I_{n}\,,
\end{equation*}
from which we easily obtain (\ref{s2 8}). It concludes the proof.
\end{proof}

\begin{corollary}The map $\Psi$ is completely $J$-positive if and only if it is a Kraus $J$-map.
\end{corollary}
\begin{proof}It is a consequence of propositions \ref{cierra1} and \ref{corolario 1 s2}.
\end{proof} \\

The following conclusive corollary represents a summary of the results of this section.

\begin{corollary}\,$\Psi$ is a quantum $J$-operator, if and only if $\Psi$ is an admissible Kraus $J$-positive map.
\end{corollary}
\begin{proof}The statement follows from theorem \ref{teorema 2 s2} and proposition \ref{corolario 1 s2}.
\end{proof} \\

Now, we relate the two sections by means of the following proposition

\begin{proposition}A quantum $J$-operator maps quantum $J$-states of unitary origin into quantum $J$-states of unitary origin.
\end{proposition}
\begin{proof}We recall that a quantum $J$-state of unitary origin $B$ satisfies the following two properties: $B$ is $J$-positive and $Tr\,BJ=1$.
Let $\Psi$ be an arbitrary quantum $J$-channel. Then, there is a matrix vector $(V_{1},\cdots,V_{\nu})$ such that, in particular, for any quantum
$J$-state $B$ of unitary origin
\begin{equation}\label{s2 9}
\Psi(B)=\sum_{i=1}^{\nu}V^{\natural}_{i}BV_{i}.
\end{equation}

Hence, it shows that $\Psi(B)$ is a $J$-positive matrix for all quantum $J$-state $B$ of unitary origin. Note also that from (\ref{s2 9}),
it follows
\begin{equation*}
\Psi(B)J=\sum_{i=1}^{\nu}(JV^{\ast}_{i}J)(BJ)(JV_{i}J)=\sum_{i=1}^{\nu}Z_{i}^{\ast}(BJ)Z_{i},
\end{equation*}
where $Z_{i}=JV_{i}J$ for $i=1,\cdots,\nu$. Define $\Phi(M)=\sum_{i=1}^{\nu}Z_{i}^{\ast}MZ_{i}$ for all $M\in M_{n}(\mathbb{C})$.
Obviously, it is a completely positive map. We claim that actually $\Phi$ is a quantum channel. Indeed, since $\Psi$ is quantum
$J$-operator, we know that
\begin{equation*}
\sum_{i=1}^{\nu}V_{i}JV_{i}^{\natural}=\sum_{i=1}^{\nu}V_{i}J(JV_{i}^{\ast}J)=J,
\end{equation*}
thus
\begin{equation*}
\sum_{i=1}^{\nu}(JV_{i}J)(JV_{i}^{\ast}J)=\sum_{i=1}^{\nu}Z_{i}Z_{i}^{\ast}=I_{n},
\end{equation*}
so $\Phi$ is a quantum channel. Note also that $\Psi(B)J=\Phi(BJ)$, it implies that $Tr\,\Psi(B)J=Tr\,\Phi(BJ)=Tr\,BJ=1$. Then,
$\Psi(B)$ is a quantum $J$-state of unitary origin.
\end{proof}

\begin{theorem}Let $\Psi$ be a completely $J$-positive map, that is, there is a matrix vector $(V_{1},\cdots,V_{l})\in M^{l}_{n}(\mathbb{C})$
such that
\begin{equation*}
\Psi(A)=\sum_{i=1}^{l}V^{\natural}_{i}AV_{i},\,\,\,\,\,\,\,\,\,\,\,\forall\,A\in M_{n}(\mathbb{C}).
\end{equation*}

Then, $\Psi$ is trace preserving, if and only if \,$\sum_{i=1}^{l}V_{i}V^{\natural}_{i}=I_{n}$.
\end{theorem}
\begin{proof}The proof is similar to the case $J=I_{n}$, that is, when $A^{\natural}=A^{\ast}$ for every $A\in M_{n}(\mathbb{C})$.
This is based on two well known facts: first, the trace is invariant under cyclic permutations and, second, $\mathcal{S}_{2}$ is
a Hilbert space with respect to the inner product $\langle A,B \rangle=Tr\,A^{\ast}B$.
\end{proof}

\begin{remark}Observe that the completely $J$-positive maps, which are trace preserving, transform quantum $J$-states of $J$-unitary origin
into quantum $J$-states of $J$-unitary origin.
\end{remark}

\section{A generalization of the previous section}

In this part of our work, results of the previous section are generalized. Moreover, the fact between the working conditions, which leads
to an essential change, is that all maps $\Psi$ considered below transform $J_{1}$-positive matrices into $J_{2}$-positive matrices, where
$J_{2}\neq \pm J_{1}$ are both fundamental symmetries of $M_{n}(\mathbb{C})$. Next, we generally keep the same notation of the previous section.

Let $J_{1}$ and $J_{2}$ be two different matrices belonging to  $M_{n}(\mathbb{C})$ which are fundamental symmetries. Thus, they introduce two different structures of indefinite metric space on $\mathbb{C}^{n}$, denoted by $(\mathbb{C}^{n}, [\cdot,\cdot]_{1})$ and $(\mathbb{C}^{n}, [\cdot,\cdot]_{2})$ respectively, where $[\cdot,\cdot]_{1}=\langle J_{1}\cdot,\cdot\rangle_{\mathbb{C}^{n}}$ and $[\cdot,\cdot]_{2}=\langle J_{2}\cdot,\cdot\rangle_{\mathbb{C}^{n}}$.
Suppose that $A$ and $B$ are two matrices such that
\begin{equation}\label{sec interc 1}
[Ax,y]_{1}=[x,By]_{2},\,\,\,\,\,\,\,\forall\,x,y\in\, \mathbb{C}^{n},
\end{equation}
then, we say that $B$ is the generalized indefinite adjoint of $A$ and it is denoted by $A^{\flat}$. Clearly, $A^{\flat}=J_{2}A^{\ast}J_{1}$.

We shall make a few of remarks about operation $\flat$.
\begin{theorem}The following statements are true
\begin{enumerate}
  \item $(A_{1}+A_{2})^{\flat}=A^{\flat}_{1}+A^{\flat}_{2}$, for all $A_{1},A_{2}\in M_{n}(\mathbb{C})$.
  \item $(\lambda A)^{\flat}=\overline{\lambda} A^{\flat}$, for any $\lambda\in \mathbb{C}$ and all $A\in M_{n}(\mathbb{C})$.
  \item $(A^{\flat})^{\flat}=A$ for all $A\in M_{n}(\mathbb{C})$ if and only if $J_{2}=\pm J_{1}$.
\end{enumerate}
\end{theorem}
\begin{proof}The first two statements are obvious. We have $(A^{\flat})^{\flat}=J_{2}J_{1}AJ_{2}J_{1}$ thus $(A^{\flat})^{\flat}=A$,
if and only if $J_{2}J_{1}=\pm I_{n}$ so, from the uniqueness of the inverse in an associative algebra, the latter could be true,
if and only if $J_{2}=\pm J_{1}$.
\end{proof}

From the previous theorem it follows that the operation $\flat$ is not an involution if $J_{2}\neq \pm J_{1}$ (which is in correspondence with our
initial assumption). Now, we introduce the concept of completely positive map in indefinite metric.

\begin{definition}We say that $\Lambda:\,(\mathbb{C}^{n}, [\cdot,\cdot]_{1})\longrightarrow (\mathbb{C}^{n}, [\cdot,\cdot]_{2})$ is
\textbf{positive in indefinite metric} if it is a linear map. Moreover, $\Lambda$ maps $J_{1}$-positive matrices into $J_{2}$-positive matrices, that is,
$\Lambda(M_{n}^{+}(\mathbb{C})(J_{1}))\subset M_{n}^{+}(\mathbb{C})(J_{2})$. It is said to be \textbf{completely positive in indefinite metric},
if $\Lambda^{k}$ is positive in indefinite metric for all $k\in \mathbb{N}$. Here, for an arbitrary linear map $\Theta$ and for all $k\in\mathbb{N}$ the $k$-th block map $\Theta^{k}$ is defined in a similar way it was done in the previous section.
\end{definition}

For a fixed $k$, that the matrix $\Lambda^{k}$ is positive in indefinite metric means for us the following  $\Lambda^{k}(M_{kn}^{+}(\mathbb{C})(\mathcal{J}^{k}_{1}))\subset M_{kn}^{+}(\mathbb{C})(\mathcal{J}^{k}_{2})$ with respect to the indefinite metrics $[\cdot,\cdot]^{k}_{i}=\langle \mathcal{J}^{k}_{i}\cdot,\cdot\rangle_{\mathbb{C}^{kn}}$ for $i=1,2$; where, as in the last section, $\mathcal{J}^{k}_{i}$
is the following block diagonal fundamental symmetry matrix of order $kn\times kn$
\begin{equation*}
\mathcal{J}^{k}_{i}=\left(
              \begin{array}{cccc}
                J_{i} & O_{n} & \cdots & O_{n} \\
                O_{n} & \ddots & \ddots & \vdots \\
                \vdots & \ddots & \ddots & O_{n} \\
                O_{n} & \cdots & O_{n} & J_{i} \\
              \end{array}
            \right),
\end{equation*}
being $i=1,2$ and $k=1,2,\cdots$.

We have
\begin{lemma}\label{lema 1 sec general}
The map $\Lambda:\,(\mathbb{C}^{n}, [\cdot,\cdot]_{1})\longrightarrow (\mathbb{C}^{n}, [\cdot,\cdot]_{2})$ is positive
in indefinite metric, if and only if there is a regular positive map $\Phi:\,M_{n}(\mathbb{C})\longrightarrow M_{n}(\mathbb{C})$ such
that $\Lambda(\cdot)=J_{2}\Phi(J_{1}\,\cdot)$.
\end{lemma}
\begin{proof}Let $\Phi$ be a usual positive map of $M_{n}(\mathbb{C})$ into $M_{n}(\mathbb{C})$ and define $\Lambda(\cdot)=J_{2}\Phi(J_{1}\,\cdot)$,
then $\Lambda$ is clearly linear. We know that $B$ is a $J_{1}$-positive matrix, if and only if $J_{1}B$ is a positive matrix, hence for all $x\in \mathbb{C}^{n}$ and any matrix $B$, which is $J_{1}$-positive
\begin{equation*}
[\Lambda(B)x,x]_{2}=\langle J_{2}\Lambda(B)x,x \rangle_{\mathbb{C}^{n}}=\langle \Phi(J_{1}B)x,x \rangle_{\mathbb{C}^{n}}\geq 0,
\end{equation*}
it shows that for all $B\in M_{n}^{+}(\mathbb{C})(J_{1})$ one has $\Lambda(B)\in M_{n}^{+}(\mathbb{C})(J_{2})$, that is, $\Lambda$ is a positive
map in indefinite metric. On the other hand, if $\Lambda$ is a positive map in indefinite metric one can prove that $\Phi(\cdot)=J_{2}\Lambda(J_{1}\,\cdot)$ is a positive map (the proof is similar to the one of lemma \ref{lema 2}) and this leads us to the other implication.
\end{proof}

If $\Psi$ is a positive map in indefinite metric, notation $\Psi\Longrightarrow \Phi$ indicates that $\Phi$ is the positive map
associated to $\Psi$ through the lemma \ref{lema 1 sec general}, that is, $\Lambda(\cdot)=J_{2}\Phi(J_{1}\,\cdot)$. The proofs of the
following facts are similar to those when $J_{2}=J_{1}$, hence they are omitted. Suppose that $\Psi\Longrightarrow \Phi$ where $\Phi$
is a completely positive map, then
\begin{enumerate}
\item $\Psi$ admits the following representation
\begin{equation}\label{sec general 1}
\Psi(A)=\sum_{i=1}^{\nu}V^{\flat}_{i}AV_{i},\,\,\,\,\,\,\forall\,A\in M_{n}(\mathbb{C}).
\end{equation}
\item $\Psi$ is completely positive in indefinite metric. This is because one has $\Psi^{k}(\cdot)=\mathcal{J}^{k}_{2}\Phi^{k}(\mathcal{J}^{k}_{1}\cdot)$ for all $k\in\mathbb{N}$.
\item Suppose that, additionally, $\Phi$ is trace preserving, that is, $\Phi$ is a usual quantum channel, then $Tr\,J_{2}\Psi(J_{1}A)=Tr\,A$.
\end{enumerate}

\begin{remark}As a consequence of the previous observations, it follows that every linear map of the form (\ref{sec general 1}) is
completely positive in indefinite metric.
\end{remark}

Next, we develop some results of Stinespring type in indefinite metric spaces. Suppose that we have a linear operator $\pi:\,M_{n}(\mathbb{C})
\longrightarrow M_{n}(\mathbb{C})$ such that $\pi(MN)=\pi(M)\pi(N)$ and $\pi(M^{\ast})=(\pi(M))^{\ast}$ for all $M,N\in M_{n}(\mathbb{C})$.
In this case, we say that $\pi$ is a representation of $M_{n}(\mathbb{C})$ in itself. Observe that $\pi$ transforms definite positive matrices into
definite positive matrices. It follows from the following fact $\pi(M^{\ast}M)=\pi(M^{\ast})\pi(M)=(\pi(M))^{\ast}\pi(M)$.

Probably, the following result belongs to the folklore of the subject, however, its proof has been included to guarantee this paper to be more
comprehensive. We have
\begin{lemma}\label{lema 2 sec general}
It turns out that $\pi^{k}:\,M_{kn}(\mathbb{C})\longrightarrow M_{kn}(\mathbb{C})$ is a positive map for all $k\in \mathbb{N}$.
\end{lemma}
\begin{proof}It is enough to prove that $\pi^{k}(\mathcal{M}^{\ast})=(\pi^{k}(\mathcal{M}))^{\ast}$ and $\pi^{k}(\mathcal{M}\mathcal{N})
=\pi^{k}(\mathcal{M})\pi^{k}(\mathcal{N})$. For two arbitrary block matrices $\mathcal{M},\mathcal{N}\in M_{kn}(\mathbb{C})$, one obtains
\begin{align*}
(\pi^{k}(\mathcal{M}\mathcal{N}))_{ij}&=\pi(\sum^{k}_{s=1}M_{is}N_{sj})=\sum^{k}_{s=1}\pi(M_{is}N_{sj})=\sum^{k}_{s=1}\pi(M_{is})\pi(N_{sj})
=\sum^{k}_{s=1}(\pi^{k}(\mathcal{M}))_{is}(\pi^{k}(\mathcal{N}))_{sj} \\
&=(\pi^{k}(\mathcal{M})\pi^{k}(\mathcal{N}))_{ij},
\end{align*}
if $1\leq i,j\leq k$. Thus, $\pi^{k}(\mathcal{M}\mathcal{N})=\pi^{k}(\mathcal{M})\pi^{k}(\mathcal{N})$. On the other hand, for $1\leq i,j\leq k$
\begin{equation*}
(\pi^{k}(\mathcal{M}^{\ast}))_{ij}=\pi(M^{\ast}_{ji})=(\pi(M_{ji}))^{\ast}=((\pi^{k}(\mathcal{M}))^{\ast})_{ij},
\end{equation*}
it shows that $\pi^{k}(\mathcal{M}^{\ast})=(\pi^{k}(\mathcal{M}))^{\ast}$.  The lemma is proved
\end{proof}

\begin{remark}Observe that we can obtain a similar result assuming that $\pi$ satisfies the following conditions $\pi(MN)=\pi(N)\pi(M)$ and $\pi(M^{\ast})=(\pi(M))^{\ast}$ for all $M,N\in M_{n}(\mathbb{C})$. In which case $\pi$ is said to be also a representation of $M_{n}(\mathbb{C})$
in itself.
\end{remark}

\begin{example}Suppose that $U$ is a unitary matrix then $\pi_{1}(A)=UAU^{\ast}$ and $\pi_{2}(A)=UA^{\ast}U^{\ast}$ are
representations of $M_{n}(\mathbb{C})$ in itself.
\end{example}

We may arrive to the following interesting result in indefinite metric spaces

\begin{theorem}Let $\pi$ be a representation of $M_{n}(\mathbb{C})$ in itself and $V\in M_{n}(\mathbb{C})$, then $\Psi(\cdot)=
J_{2}V^{\ast}\pi(J_{1}\cdot)V$ is completely positive in indefinite metric.
\end{theorem}
\begin{proof}First of all, observe that for any $k\in\mathbb{N}$ and all block matrix $\mathcal{C}$
of order $kn\times kn$, we have
\begin{equation*}
\Psi^{k}(\mathcal{C})=\mathcal{J}^{k}_{2}\mathcal{V}_{k}^{\ast}\pi^{k}(\mathcal{J}^{k}_{1}\mathcal{C})\mathcal{V}_{k}\,,
\end{equation*}
where, as above,
\begin{equation*}
\mathcal{V}_{k}=\left(
              \begin{array}{cccc}
                V & O_{n} & \cdots & O_{n} \\
                O_{n} & \ddots & \ddots & \vdots \\
                \vdots & \ddots & \ddots & O_{n} \\
                O_{n} & \cdots & O_{n} & V \\
              \end{array}
            \right),\,\,\,\,\,\,\,\,\,\,\,\,\,\mathcal{V}^{\ast}_{k}=\left(
              \begin{array}{cccc}
                V^{\ast} & O_{n} & \cdots & O_{n} \\
                O_{n} & \ddots & \ddots & \vdots \\
                \vdots & \ddots & \ddots & O_{n} \\
                O_{n} & \cdots & O_{n} & V^{\ast} \\
              \end{array}
            \right),
\end{equation*}
are block diagonal matrices of order $kn\times kn$. We claim that $\Psi^{k}(\cdot):\,(\mathbb{C}^{kn},[\cdot,\cdot]_{1}^{k})\longrightarrow
(\mathbb{C}^{kn},[\cdot,\cdot]_{2}^{k})$ is positive in indefinite metric for all $1\leq k$. In fact, let $k$ be a fixed positive integer but
arbitrary. Then, from the lemma \ref{lema 2 sec general} it follows that  $\Phi_{k}(\cdot)=\mathcal{V}_{k}^{\ast}\pi^{k}(\cdot)\mathcal{V}_{k}$
is a usual positive map on $M_{kn}(\mathbb{C})$. On the other hand, $\Psi^{k}(\cdot)=\mathcal{J}_{2}^{k}\Phi_{k}(\mathcal{J}_{1}^{k}\cdot)$, which
implies that $\Psi^{k}$ is positive in indefinite metric for all $1\leq k$ due to lemma \ref{lema 1 sec general}. Thus, the map $\Psi$ is completely
positive in indefinite metric.
\end{proof}

As above, suppose that $J_{1}$ and $J_{2}$ are fundamental symmetries. For $i=1,2$, define $A^{\natural_{i}}=J_{i}A^{\ast}J_{i}$ for all $A\in M_{n}(\mathbb{C})$. Then, as it was already mentioned the operations $\natural_{i}$ are involutions on $M_{n}(\mathbb{C})$.

\begin{definition}We say that $\pi\,:M_{n}(\mathbb{C})\longrightarrow M_{n}(\mathbb{C})$ is a \textbf{representation in indefinite metric of $M_{n}(\mathbb{C})$ in itself}, if it is a linear map, such that $\pi(AJ_{1}B)=\pi(A)J_{2}\pi(B)$ and $\pi(A^{\natural_{1}})=(\pi(A))^{\natural_{2}}$.
\end{definition}

We give an example

\begin{example}Suppose that $\pi:\,M_{n}(\mathbb{C})\longrightarrow M_{n}(\mathbb{C})$ is a linear map such that
\begin{equation*}
\pi(AB)=\pi(A)\pi(B),\,\,\,\,\,\,\,\,\,\,\,\,\,\pi(A^{\ast})=(\pi(A))^{\ast},\,\,\,\,\,\,\,\,\,\,\pi(J_{1})=J_{2},\,\,\,\,\,\,\,\,
\forall\, A,B\in M_{n}(\mathbb{C}),
\end{equation*}
then $\pi$ is a representation in indefinite metric of $M_{n}(\mathbb{C})$ in itself. In fact, we have $\pi(AJ_{1}B)=\pi(A)\pi(J_{1})\pi(B)
=\pi(A)J_{2}\pi(B)$ for all $A,B\in M_{n}(\mathbb{C})$. On the other hand, $\pi(A^{\natural_{1}})=\pi(J_{1}A^{\ast}J_{1})=\pi(J_{1})
\pi(A^{\ast})\pi(J_{1})=J_{2}(\pi(A))^{\ast}J_{2}=(\pi(A))^{\natural_{2}}$ for any $A\in M_{n}(\mathbb{C})$.
\end{example}

\begin{remark}\label{haciajordan1}
Observe that if $\pi$ is a representation in indefinite metric of $M_{n}(\mathbb{C})$ in itself, then it is positive
in indefinite metric. Indeed, according to the previous definition, we have $\pi(B^{\natural_{1}}J_{1}B)=\pi(B^{\natural_{1}})J_{2}\pi(B)
=(\pi(B))^{\natural_{2}}J_{2}\pi(B)$ for all $B\in M_{n}(\mathbb{C})$. On the other hand, if $A=A^{\natural_{1}}$, then, $\pi(A)=\pi(A^{\natural_{1}})
=(\pi(A))^{\natural_{2}}$. It implies that $\pi$ maps $J_{1}$-positive matrices into $J_{2}$-positive matrices.
\end{remark}

We have
\begin{theorem}Any representation $\pi$ in indefinite metric of $M_{n}(\mathbb{C})$ in itself is a completely positive map in indefinite metric.
\end{theorem}

\begin{proof}We should prove that for each $k\in\mathbb{Z}_{+}$ the map $\pi^{k}$ satisfies the following properties $\pi^{k}(\mathcal{C}\mathcal{J}^{k}_{1}\mathcal{D})=\pi^{k}(\mathcal{C})\mathcal{J}^{k}_{2}\pi^{k}(\mathcal{D})$ and
$\pi^{k}(\mathcal{C}^{\natural(\mathcal{J}^{k}_{1})})=(\pi^{k}(\mathcal{C}))^{\natural(\mathcal{J}^{k}_{2})}$ for all
$\mathcal{C},\mathcal{D}\in M_{kn}(\mathbb{C})$ where $\mathcal{C}^{\natural(\mathcal{J}^{k}_{i})})=\mathcal{J}^{k}_{i}\mathcal{C}^{\ast}\mathcal{J}^{k}_{i}$ for $i=1,2$. Consider
a fixed $k$ but arbitrary, if $1\leq i,j \leq k$ then, for two block matrices $\mathcal{M},\mathcal{N}\in M_{kn}(\mathbb{C})$ we have
\begin{align*}
(\pi^{k}(\mathcal{M}\mathcal{J}_{1}^{k}\mathcal{N}))_{ij}&=\pi(\sum^{k}_{s=1}M_{is}J_{1}N_{sj})=\sum^{k}_{s=1}\pi(M_{is}J_{1}N_{sj})
=\sum^{k}_{s=1}\pi(M_{is})J_{2}\pi(N_{sj}) \\
&=\sum^{k}_{s=1}(\pi^{k}(\mathcal{M}))_{is}J_{2}(\pi^{k}(\mathcal{N}))_{sj}
=(\pi^{k}(\mathcal{M})\mathcal{J}_{2}^{k}\pi^{k}(\mathcal{N}))_{ij},
\end{align*}
which implies that $\pi^{k}(\mathcal{M}\mathcal{J}_{1}^{k}\mathcal{N})=\pi^{k}(\mathcal{M})\mathcal{J}_{2}^{k}\pi^{k}(\mathcal{N})$. Moreover,
\begin{align*}
(\pi^{k}(\mathcal{M}^{\natural(\mathcal{J}_{1}^{k})}))_{ij}&=(\pi^{k}(\mathcal{J}_{1}^{k}\mathcal{M}^{\ast}\mathcal{J}_{1}^{k})_{ij}=
\pi(J_{1}M^{\ast}_{ji}J_{1})=\pi(M^{\natural_{1}}_{ji})=(\pi(M_{ji}))^{\natural_{2}}=J_{2}(\pi(M_{ji}))^{\ast}J_{2}  \\
&=J_{2}((\pi^{k}(\mathcal{M}))^{\ast})_{ij}J_{2}=(\mathcal{J}_{2}^{k}(\pi^{k}(\mathcal{M}))^{\ast}\mathcal{J}_{2}^{k})_{ij}
=\left ((\pi^{k}(\mathcal{M}))^{\natural(\mathcal{J}_{2}^{k})}\right)_{ij},
\end{align*}
thus $\pi^{k}(\mathcal{M}^{\natural(\mathcal{J}_{1}^{k})})=(\pi^{k}(\mathcal{M}))^{\natural(\mathcal{J}_{2}^{k})}$ for all
$\mathcal{M}\in M_{n}(\mathbb{C})$. The theorem follows from remark \ref{haciajordan1}.
\end{proof}

\begin{definition}We say that a linear map $\Psi$ is a quantum $(J_{1},J_{2})$-channel if $\Psi$ can be represented in the form (\ref{sec general 1})
and moreover it is trace preserving.
\end{definition}

\begin{theorem}A linear map $\Psi$ is a quantum $(J_{1},J_{2})$-channel, if and only if it has the form (\ref{sec general 1}) and
the matrix vector $(V_{1},\cdots,V_{\nu})$ which arises from this representation satisfies the property
\begin{equation*}
\sum_{i=1}^{\nu}V^{\flat}_{i}V_{i}=I_{n}\,.
\end{equation*}
\end{theorem}

\section{Conclusions}

To conclude this paper, two open problems are being proposed, in which the authors of the present article are currently making progress\,:
\begin{enumerate}

  \item The study of the dynamics of linear operators is an important and popular topic. For example, in this sense, the reader can browse
        a classic reference, like book \cite{bayart}. In particular, from this point of view, completely positive maps, have not been exempt
        from analysis . In this direction, they were object of interest in many works. For example, \cite{raha}, \cite{carbone} and
        references therein. Specifically, about $15$ years ago, the study of quantum channels fixed points was subjected to strong research.
        On the other hand, it is well known that the multiplicative properties of quantum channels play a central role in order to obtain these
        results. We suggest to study the dynamics of the quantum channels defined on spaces with an indefinite metric.

  \item Inspired by the first problem, we consider that it would be interesting and feasible to reveal the relationship between quantum
        $J$-channels and the Jordan algebras (for the usual case $J=I_{n}$, we suggest to see the book \cite{stormer}).
\end{enumerate}

\section{Acknowledgment}

Ra\'{u}l Felipe was supported in part under CONACYT grant $45886$. The authors would like to thank the referees suggestions
to improve our article.

\end{document}